\documentclass[sigconf]{aamas_hacked}

\usepackage{thm-restate}
\usepackage{makecell}
\usepackage{cleveref}
\usepackage{amsthm}
\usepackage{bm}
\usepackage{enumitem}
\usepackage{microtype}
\usepackage[english]{babel}

\newtheoremstyle{bfnote}
{}{}
{}{}
{\bfseries}{.}
{ }
{\thmname{#1}\thmnumber{ #2}\thmnote{ (#3)}}
\theoremstyle{bfnote}
\newtheorem{remark}{Remark}

\newcommand{\sd}{\mathit{SD}}
\newenvironment{profile}{\medmuskip=0mu\relax
	\thickmuskip=1mu\relax
	\tabular}{\endtabular\smallskip}

\newcommand{\ssucc}{\vartriangleright}

\newcolumntype{L}[1]{>{\raggedright\let\newline\\\arraybackslash\hspace{0pt}}m{#1}}
\newcolumntype{C}[1]{>{\centering\let\newline\\\arraybackslash\hspace{0pt}}m{#1}}
\newcolumntype{R}[1]{>{\raggedleft\let\newline\\\arraybackslash\hspace{0pt}}m{#1}}

\theoremstyle{plain}
\newtheorem{lemma}{Lemma}

\renewenvironment{proof}{\noindent\textbf{Proof.}}{\hfill\qed\medskip}
\newenvironment{proofsketch}{\noindent\textbf{Proof sketch.}}{\hfill\qed\medskip}

\title[Strategyproof Social Decision Schemes on Super Condorcet Domains]{Strategyproof Social Decision Schemes\\ on Super Condorcet Domains}

\author{Felix Brandt}
\affiliation{
  \institution{Technical University of Munich}
  \city{Munich}
  \country{Germany}}
\email{brandtf@in.tum.de}

\author{Patrick Lederer}
\affiliation{
  \institution{Technical University of Munich}
  \city{Munich}
  \country{Germany}}
\email{ledererp@in.tum.de}

\author{Sascha Tausch}
\affiliation{
  \institution{Technical University of Munich}
  \city{Munich}
  \country{Germany}}
\email{sascha.tausch@tum.de}

\begin{abstract}
One of the central economic paradigms in multi-agent systems is that agents should not be better off by acting dishonestly. In the context of collective decision-making, this axiom is known as strategyproofness and turns out to be rather prohibitive, even when allowing for randomization. In particular, Gibbard's random dictatorship theorem shows that only rather unattractive social decision schemes (SDSs) satisfy strategyproofness on the full domain of preferences. In this paper, we obtain more positive results by investigating strategyproof SDSs on the Condorcet domain, which consists of all preference profiles that admit a Condorcet winner. In more detail, we show that, if the number of voters $n$ is odd, every strategyproof and non-imposing SDS on the Condorcet domain can be represented as a mixture of dictatorial SDSs and the Condorcet rule (which chooses the Condorcet winner with probability $1$). Moreover, we prove that the Condorcet domain is a maximal connected domain that allows for attractive strategyproof SDSs if $n$ is odd as only random dictatorships are strategyproof and non-imposing on any sufficiently connected superset of it.
We also derive analogous results for even $n$ by slightly extending the Condorcet domain. 
Finally, we also characterize the set of group-strate\-gyproof and non-imposing SDSs on the Condorcet domain and its supersets. These characterizations strengthen Gibbard's random dictatorship theorem and establish that the Condorcet domain is essentially a maximal domain that allows for attractive strategyproof SDSs.
\end{abstract}

\allowdisplaybreaks

\begin{document}

	\pagestyle{fancy}
	\fancyhead{}
	
\maketitle 

\section{Introduction}

Strategyproofness---no agent should be better of by acting dishon\-estly---is one of the central economic paradigms in multi-agent systems \cite[][]{NRTV07a,ShLe08a,BCE+14a}.
An important challenge for such systems is the identification of socially desirable outcomes by letting the agents cast votes that represent their preferences over the possible alternatives. A multitude of theorems in economic theory have shown that even rather basic properties of voting rules cannot be satisfied simultaneously.
In this context, strategyproofness is known to be a particularly restrictive axiom. This is exemplified by the Gibbard- atterthwaite theorem which states that dictatorships are the only deterministic voting rules that satisfy strategyproofness and non-imposition (i.e., every alternative is elected for some preference profile). Since dictatorships are not acceptable for most applications, this result is commonly considered an impossibility theorem. 

One of the most successful escape routes from the Gibbard- Satterthwaite impossibility is to restrict the domain of feasible preference profiles. For instance, \citet{Moul80a} prominently showed that there are attractive strategyproof voting rules on the domain of single-peaked preference profiles, and various other restricted domains of preferences have been considered since then \citep[e.g.,][]{BGS93a,PeSt99a,Sapo09a,ELP17a}. The idea behind domain restrictions is that the voters' preferences often obey structural constraints and thus, not all preference profiles are likely or plausible. A particularly significant constraint is the existence of a Condorcet winner which is an alternative that is favored to every other alternative by a majority of the voters. Apart from its natural appeal, this concept is important because there is strong empirical evidence that real-world elections usually admit Condorcet winners \citep[][]{RGMT06a,Lasl10a,GeLe11a}. This motivates the study of the Condorcet domain which consists precisely of the preference profiles that admit a Condorcet winner.
Note that the Condorcet domain is a superset of several important domains such as those of single-peaked and single-dipped preferences when the number of voters is odd. There are several results showing the existence of attractive strategyproof voting rules on the Condorcet domain.
In particular, \citet{CaKe03a} characterize the Condorcet rule, which always picks the Condorcet winner, as the only strategyproof, non-imposing, and non-dictatorial voting rule on the Condorcet domain if the number of voters is odd. 

In this paper, we focus on \emph{randomized} voting rules, so-called social decision schemes (SDSs). \citet{Gibb77a} has shown that randomization unfortunately does not allow for much more leeway beyond the negative consequences of the Gibbard-Satterthwaite theorem: random dictatorships, which select each voter with a fixed probability and elect the favorite alternative of the chosen voter, are the only SDSs on the full domain that satisfy strategyproofness and non-imposition (which in the randomized setting requires that every alternative is chosen with probability $1$ for some preference profile). Thus, these SDSs are merely ``mixtures of dictatorships''. 

{\fussy
In order to circumvent this negative result, we are interested in large domains that allow for strategyproof and non-imposing SDSs apart from random dictatorships. A natural candidate for this is the Condorcet domain and, indeed, we show that the Condorcet domain is essentially a maximal domain that allows for strategyproof, non-imposing, and ``non-randomly dictatorial'' social choice. In more detail, we prove that, if the number of voters $n$ is odd, every strategyproof and non-imposing SDS on the Condorcet domain can be represented as a mixture of dictatorial SDSs and the Condorcet rule (which chooses the Condorcet winner with probability $1$). This result entails that the Condorcet rule is the only strategyproof, non-imposing, and completely ``non-randomly dictatorial'' SDS on the Condorcet domain for odd $n$. Moreover, we show that, if $n$ is odd, the Condorcet domain is a maximal domain that allows for strategyproof and non-imposing SDSs other than random dictatorships. This theorem highlights the importance of Condorcet winners for the existence of attractive strategyproof SDSs.}

Unfortunately, our results for the Condorcet domain fail if the number of voters $n$ is even because, in this case, a single voter cannot change the Condorcet winner. For extending our results to an even number of voters, we consider tie-breaking Condorcet domains, which contain all preference profiles that have a Condorcet winner after majority ties are broken according to a fixed tie-breaking order. Tie-breaking Condorcet domains are supersets of the Condorcet domain for even $n$, and we derive analogous results for these domains as for the Condorcet domain: if $n$ is even, only mixtures of random dictatorships and the tie-breaking Condorcet rule (which chooses the Condorcet winner after the majority ties have been broken) are strategyproof and non-imposing on these domains, and only random dictatorships satisfy these properties on connected supersets. Finally, we also characterize the set of group-strategyproof and non-imposing SDSs on the Condorcet domain and most of its supersets independently of the parity of $n$: while the Condorcet rule satisfies these axioms on the Condorcet domain, only dictatorships are able to do so on most of its superdomains.

\sloppy
In summary, our results demonstrate two important insights: \emph{(i)} the Condorcet domain is essentially a maximal domain that allows for strategyproof, non-randomly dictatorial, and non-imposing SDSs, and \emph{(ii)} the (deterministic) Condorcet rule is the most appealing strategyproof voting rule on this domain, even if we allow for randomization. Our characterizations can also be seen as attractive complements to classic negative results for the full domain, whereas our results for supersets of the (tie-breaking) Condorcet domain significantly strengthen these negative results. In particular, our theorems imply statements by \citet{Barb79a} and \citet{CaKe03a} as well as the Gibbard-Satterthwaite theorem \citep{Gibb73a,Satt75a} and the random dictatorship theorem \citep{Gibb77a}. A more detailed comparison between our results to these classic theorems is given in \Cref{tab:FDvsCD}.

\section{Related Work}

Restricting the domain of preference profiles in order to circumvent classic impossibility theorems has a long tradition and remains an active research area to date. In particular, the existence of attractive deterministic voting rules that satisfy strategyproofness has been shown for a number of domains. Classic examples include the domains of single-peaked \citep{Moul80a}, single-dipped \citep{BBM12b}, and single-crossing \citep{Sapo09a} preference profiles. More recent positive results focus on broader but more technical domains such as the domains of multi-dimensionally single-peaked or semi single-peaked preference profiles \citep[e.g.,][]{BGS93a,NePu03a,CSS13a,Reff15a}. On the other hand, domain restrictions are also used to strengthen impossibility results by proving them for smaller domains \citep[e.g.,][]{ACS03a,Sato10a,GoRo18a}. In more recent research, the possibility and impossibility results converge by giving precise conditions under which a domain allows for strategyproof and non-dictatorial deterministic voting rules \citep[][]{ChSe11a,CSS13a,RoSt19a,ChZe21a}. 

While similar results have also been put forward for SDSs, this setting is not as well understood. For instance, \citet{EPS02a} have shown the existence of attractive strategyproof SDSs on the domain of single-peaked preference profiles \citep[see also][]{PRSS14a,PyUn15a}. 
The existence of strategyproof and non-\-im\-po\-sing SDSs other than random dictatorships has also been investigated for a variety of other domains \citep[][]{PRSS17a,RoSa20a,PRS19a}. Following a more general approach, \citet{CSZ14a} and \citet{ChZe18a} identify criteria for deciding whether a domain admits such SDSs.

The strong interest in restricted domains also led to the study of many computational problems for restricted domains \citep[e.g.,][]{Coni09b,FHHR11a,BCW13a,BBHH15a,ELP16a,Pete17c,PeLa20a}. For instance, \citet{BCW13a} give an algorithm for recognizing whether a preference profile is single-crossing. Note that for the Condorcet domain, this problem can be solved efficiently as it is easy to verify the existence of a Condorcet winner.

Finally, observe that all aforementioned results are restricted to \emph{Cartesian} domains, i.e., domains of the form $\mathcal{D}=\mathcal{X}^N$, where $\mathcal{X}$ is a set of preference relations. However, the Condorcet domain is not Cartesian. In this sense, the only results directly related to ours are the ones by \citet{CaKe03a} and their follow-up work \citep{Merr11a,CaKe15a,CaKe16a}. These papers can be seen as predecessors of our work since they investigate strategyproof deterministic voting rules on the Condorcet domain. In particular, our results extend the results by \citet{CaKe03a} in multiple important ways: we allow for randomization, we explore the case of even $n$ by slightly extending the domain, we demonstrate the boundary of the possibility results, and we analyze the consequences of group-strategyproofness.

\section{Preliminaries}

Let $N=\{1,\dots, n\}$ denote a finite set of voters and $A=\{a,b,\dots\}$ be a finite set of $m$ alternatives. Throughout the paper, we assume that there are $n\geq 3$ voters and $m\geq 3$ alternatives. Every voter $i\in N$ is equipped with a \emph{preference relation} $\succ_i$ which is a complete, transitive, and anti-symmetric binary relation on $A$. We define $\mathcal{R}$ as the set of all preference relations on $A$. A \emph{preference profile} $R\in \mathcal{R}^N$ consists of the preference relations of all voters $i\in N$. A domain of preference profiles $\mathcal{D}$ is a subset of the full domain $\mathcal{R}^N$. When writing preference profiles, we represent preference relations as comma-separated lists and indicate the set of voters who share a preference relation directly before the preference relation. Finally, we use ``$\dots$'' to indicate that the missing alternatives can be ordered arbitrarily. 
For instance, $\{1,2\}\colon a,b,c,\dots$ means that voters $1$ and $2$ prefer $a$ to $b$ to $c$ to all remaining alternatives, which can be ordered arbitrarily. We omit the brackets for singleton sets. 

The main object of study in this paper are \emph{social decision schemes (SDSs)} which are voting rules that may use randomization to determine the winner of an election. More formally, an SDS maps every preference profile $R$ of a domain $\mathcal{D}$ to a lottery over the alternatives that determines the winning chance of every alternative. A \emph{lottery} $p$ is a probability distribution over the alternatives, i.e., $p(x)\geq 0$ for all $x\in A$ and $\sum_{x\in A} p(x)=1$. We define $\Delta(A)$ as the set of all lotteries over $A$. Formally, an SDS on a domain $\mathcal{D}$ is then a function of the type $f:\mathcal{D}\rightarrow \Delta(A)$. Hence, SDSs are a generalization of deterministic voting rules which choose an alternative with probability $1$ in every preference profile. The term $f(R,x)$ denotes the probability assigned to $x$ by the lottery $f(R)$. For every set $X\subseteq A$ and lottery $p$, we define $p(X)=\sum_{x\in A} p(x)$; in particular $f(R,X)=\sum_{x\in X} f(R,x)$. Finally, an SDS $f:\mathcal{D}\rightarrow\Delta(A)$ is a \emph{mixture of SDSs} $g_1,\dots, g_k$ if there are values $\lambda_i\geq 0$ for $i\in \{1,\dots, k\}$ such that $f(R)=\sum_{i=1}^k\lambda_i g_i(R)$ for all profiles $R\in\mathcal{D}$.

A natural desideratum for an SDS $f:\mathcal{D}\rightarrow\Delta(A)$ is \emph{non-impo\-si\-tion} which requires that there is for every alternative $x\in A$ a profile $R\in\mathcal{D}$ such that $f(R,x)=1$. A prominent strengthening of this property is \emph{ex post} efficiency. For defining this axiom, we say an alternative $x\in A$ Pareto-dominates another alternative $y\in A\setminus \{x\}$ in a profile $R$ if $x \succ_i y$ for all voters $i\in N$. Then, an SDS $f:\mathcal{D}\rightarrow\Delta(A)$ is \emph{ex post efficient} if $f(R,x)=0$ for all alternatives $x\in A$ and profiles $R\in\mathcal{D}$ such that $x$ is Pareto-dominated in $R$.

\subsection{Strategyproofness and Random Dictatorships}

Strategic manipulation is one of the central issues in social choice theory: voters might be better off by voting dishonestly. Since satisfactory collective decisions require the voters' true preferences, SDSs should incentivize honest voting. In order to formalize this, we need to specify how voters compare lotteries over alternatives. The most prominent approach for this is based on (first order) stochastic dominance \citep[e.g.,][]{Gibb77a,EPS02a,PRS19a}. Let the \emph{upper contour set} $U(\succ_i,x)=\{y\in A\colon y \succ_i x\lor y=x\}$ be the set of alternatives that voter $i$ weakly prefers to $x$. Then, (first order) stochastic dominance states that a voter $i$ prefers a lottery $p$ to another lottery $q$, denoted by $p\succsim_i^\sd q$, if $p(U(\succ_i,x)) \geq q(U(\succ_i,x))$ for all $x\in A$. Note that the stochastic dominance relation is transitive but not complete. Using stochastic dominance to compare lotteries is appealing because $p\succsim_i^\sd q$ holds if and only if $p$ guarantees voter $i$ at least as much expected utility than $q$ for every utility function that is ordinally consistent with his preference relation $\succ_i$.

Based on stochastic dominance, we now define strategyproofness: an SDS $f:\mathcal{D}\rightarrow\Delta(A)$ is \emph{strategyproof} if $f(R)\succsim_i^\sd f(R')$ for all preference profiles $R,R'\in\mathcal{D}$ and voters $i\in N$ such that ${\succ_j}={\succ_j'}$ for all $j\in N\setminus \{i\}$. Less formally, strategyproofness requires that every voter weakly prefers the lottery obtained by acting truthfully to every lottery he could obtain by lying. Conversely, an SDS is called \emph{manipulable} if it is not strategyproof. A convenient property of strategyproofness is that mixtures of strategyproof SDSs are again strategyproof. 

Note that this strategyproof notion has attained significant attention. In particular, \citet{Gibb77a} has shown that only random dictatorships satisfy strategyproofness and non-imposition on the full domain. For defining these functions, we say an SDS $d_i$ is \emph{dictatorial} or a \emph{dictatorship} if it always assigns probability $1$ to the most preferred alternative of voter $i$. Then, a \emph{random dictatorship} $f$ is a mixture of dictatorial SDSs $d_i$. While they are more attractive than dictatorships, random dictatorships are often undesirable because they cannot compromise. We therefore interpret the random dictatorship theorem as a negative result.

Furthermore, strategyproofness is closely related to two properties called localizedness and non-perversity. Both of these axioms are concerned with how the outcome changes if a voter only swaps two alternatives. For making this formal, let $R^{i:yx}$ denote the profile derived from another profile $R$ by only swapping $x$ and $y$ in the preference relation of voter $i$. Note that this definition requires that $x\succ_i y$ and that there is no alternative $z\in A\setminus \{x,y\}$ with $x \succ_i z \succ_i y$. Now, an SDS $f$ on a domain $\mathcal{D}$ is \emph{localized} if $f(R,z)=f(R^{i:yx},z)$ for all distinct alternatives $x,y,z\in A$, voters $i\in N$, and profiles $R, R^{i:yx}\in\mathcal{D}$. Moreover, $f$ is \emph{non-perverse} if $f(R^{i:yx},y)\geq f(R, y)$ for all distinct alternatives $x,y\in A$, voters $i\in N$, and profiles $R,R^{i:yx}\in \mathcal{D}$. More intuitively, if voter $i$ reinforces $y$ against $x$, localizedness requires that the probability assigned to the other alternatives does not change, and non-perversity that the probability of $y$ cannot decrease. \citet{Gibb77a} has shown for the full domain of preferences that the conjunction of localizedness and non-perversity is equivalent to strategyproofness. Furthermore, it is easy to see that every strategyproof SDS satisfies non-perversity and localizedness on every domain. We thus mainly use the latter two axioms in our proofs as they are easier to handle.

Finally, in order to disincentivize \emph{groups} of voters from manipulating, we need a stronger strategyproofness notion: an SDS $f:\mathcal{D}\rightarrow\Delta(A)$ is \emph{group-strategyproof} if for all preference profiles $R,R'\in\mathcal{D}$ and all non-empty sets of voters $I\subseteq N$ with ${\succ_j}={\succ_j'}$ for all $j\in N\setminus I$, there is a voter $i\in I$ such that $f(R)\succsim_i^\sd f(R')$. Conversely, an SDS is \emph{group-manipulable} if it is not group-strategyproof. Note that group-strategyproofness implies strategyproofness.  

\subsection{Super Condorcet Domains}

Since Gibbard's random dictatorship theorem shows that there are no attractive strategyproof SDSs on the full domain, we investigate the Condorcet domain and its supersets with respect to the existence of such functions. In order to define these domains, we first need to introduce some terminology. The \emph{majority margin} $g_{R}(x,y)=|\{i\in N\colon x \succ_i y\}|-|\{i\in N\colon y \succ_i x\}|$ indicates how many more voters prefer $x$ to $y$ in the profile $R$ than vice versa. Based on the majority margins, we define the \emph{Condorcet winner} of a profile $R$ as the alternative $x$ such that $g_{R}(x,y)>0$ for all $y\in A\setminus \{x\}$. Since the existence of a Condorcet winner is not guaranteed, we focus on the \emph{Condorcet domain} $\mathcal{D}_C=\{R\in \mathcal{R}^N\colon \text{there is a Condorcet winner in }R\}$ which contains for the given electorate all profiles with a Condorcet winner. As explained in the introduction, this domain is of interest because real-world elections frequently admit a Condorcet winner. 

A particularly natural SDS on the Condorcet domain is the \emph{Condorcet rule ($\mathit{COND}$)} which always assigns probability $1$ to the Condorcet winner. However, all SDSs defined for the full domain (e.g., random dictatorships, Borda's rule, Plurality rule) are also well-defined for the Condorcet domain and there is thus a multitude of voting rules to choose from.

Note that the Condorcet domain $\mathcal{D}_C$ is not connected with respect to strategyproofness if $n$ is even. To make this formal, we define $\mathcal{D}_C^x$ as the domain of profiles in which alternative $x$ is the Condorcet winner. Then, it is impossible for distinct alternatives $x,y\in A$ that a single voter deviates from a profile $R\in\mathcal{D}_C^x$ to a profile $R'\in \mathcal{D}_C^y$. Indeed, if $R\in\mathcal{D}_C^x$ and $n$ is even, then $g_R(x,y)\geq 2$ and $g_{R'}(x,y)\geq 0$ for all alternatives $y\in A\setminus \{x\}$ and profiles $R'$ that differ from $R$ only in the preference relation of a single voter. This is problematic for our analysis because the choice for a profile $R\in\mathcal{D}_C^x$ has no influence of the choice for a profile $R'\not\in\mathcal{D}_C^x$. When $n$ is even, we will therefore consider the \emph{tie-breaking Condorcet domain} $\mathcal{D}_C^{\ssucc}=\{R\in\mathcal{R}^N\colon \text{there is a Condorcet winner in } (R,\ssucc)\}$, which contains all profiles that have a Condorcet winner after adding a fixed preference relation ${\ssucc}\in \mathcal{R}$. Note that this extra preference relation only breaks majority ties if $n$ is even because $g_R(x,y)\geq 2$ if $g_R(x,y)\neq 0$. In particular, this proves that $\mathcal{D}_C\subseteq\mathcal{D}_C^\ssucc$ if $n$ is even.
An attractive SDS on $\mathcal{D}_C^{\ssucc}$ is the \emph{tie-breaking Condorcet rule ($\mathit{COND}^{\ssucc}$)} which assigns probability $1$ to the Condorcet winner in $(R,\ssucc)$ for all profiles $R\in\mathcal{D}_C^{\ssucc}$.

To show that $\mathcal{D}_C$ and $\mathcal{D}_C^{\ssucc}$ are maximal domains that allow for attractive strategyproof SDSs, we will also consider supersets of them. Formally, we analyze \emph{super Condorcet domains} which are domains $\mathcal{D}$ with $\mathcal{D}_C\subseteq \mathcal{D}$. Just as the Condorcet domain for even $n$, super Condorcet domains can be disconnected. We therefore discuss connectedness notions for domains and introduce ad-paths. An \emph{ad-path} from a profile $R$ to a profile $R'$ in a domain $\mathcal{D}$ is a sequence of profiles $(R^1,\dots, R^l)$ such that $R^1=R$, $R^l=R'$, $R^k\in \mathcal{D}$ for all $k\in \{1,\dots, l\}$, and the profile $R^{k+1}$ evolves out of $R^k$ by swapping two alternatives $x,y\in A$ in the preference relation of a voter $i\in N$, i.e., $R^{k+1}=(R^k)^{i:yx}$ for all $k\in \{1,\dots, l-1\}$. Then, we say that a domain $\mathcal{D}$ is \emph{weakly connected} if there is an ad-path between all profiles $R,R'\in\mathcal{D}$. 
Unfortunately, this condition is too weak to be useful in our analysis and we therefore slightly strengthen it: 
a domain $\mathcal{D}$ is \emph{connected} if it is weakly connected and if for all alternatives $x\in A$ and profiles $R,R'\in\mathcal{D}$ such that $U(\succ_i',x)= U(\succ_i,x)$ for all $i\in N$, there is an ad-path $(R^1, \dots, R^l)$ from $R$ to $R'$ such that $U(\succ_i^{k},x)= U(\succ_i^{k-1},x)$ for all $i\in N$ and $k\in \{2,\dots,l\}$. Less formally, connectedness strengthens weak connectedness by requiring that if an alternative $x$ is at the same position in $R$ and $R'$, then there is an ad-path from $R$ to $R'$ along which $x$ is not moved.

Connectedness is a very mild property and is, e.g., weaker than \citeauthor{Sato13b}'s non-restoration property \citep{Sato13b}. Consequently, many domains, such as the full domain and the domain of single-peaked preferences, satisfy this condition. As we show next, the Condorcet domain is connected for odd $n$ and tie-breaking Condorcet domains are connected for even $n$. 

\begin{restatable}{lemma}{connected}\label{lem:connected}
If $n\geq 3$ is odd, the Condorcet domain $\mathcal{D}_C$ is connected. If $n\geq 4$ is even, the tie-breaking Condorcet domain $\mathcal{D}_C^\ssucc$ is connected for every preference relation ${\ssucc}\in\mathcal{R}$.
\end{restatable}
\begin{proofsketch}
The proof for $\mathcal{D}_C$ and $\mathcal{D}_C^\ssucc$ work essentially the same, and we thus focus on $\mathcal{D}_C$ in this proof sketch. Hence, assume that $n\geq 3$ is odd, consider two profiles $R, R'\in\mathcal{D}_C$, and let $c$ and $c'$ be the respective Condorcet winners. We first show that $\mathcal{D}_C$ is weakly connected and thus need to construct an ad-path from $R$ to $R'$. If $c=c'$, we can simply start at $R$ by reinforcing $c$ until it unanimously top-ranked, reorder the remaining alternatives according to $R'$, and weaken $c$ until it is in the correct position. If $c\neq c'$, we can proceed similarly: starting at $R$, we let all voters first push up $c$ until it is their best alternative, and then let all voters push up $c'$ until it is their second best alternative. We can now change the Condorcet winner without leaving the Condorcet domain by letting the voters swap $c$ and $c'$ one after another. After this, $c'$ is the Condorcet winner and we can now apply the same construction as for the case $c=c'$ to go from this intermediate profile to $R'$. For showing that $\mathcal{D}_C$ is connected, we also need to construct ad-paths from $R$ to $R'$ that do not move $x$ for every alternative $x\in A$ with $U(\succ_i,x)=U(\succ_i',x)$ for all $i\in N$. Since the construction of these ad-paths requires a case distinction with respect to whether $x=c$ and $c=c'$, we defer it to the appendix.
\end{proofsketch}
    
\section{Characterizations}

We are now ready to present our characterizations of strategyproof and group-strategyproof SDSs on super Condorcet domains. In more detail, we first characterize the set of strategyproof and non-imposing SDSs on the Condorcet domain for and odd number of voters $n$ in \Cref{subsec:cond}. Moreover, we also demonstrate that the Condorcet domain is a maximal connected domain that allows for strategyproof and non-imposing SDSs apart from random dictatorships. Next, we derive analogous results for the tie-breaking Condorcet domain if $n$ is even in \Cref{subsec:condeven}. Finally, in \Cref{subsec:GSP} we revisit the Condorcet domain and characterize the set of group-strategyproof and non-imposing SDSs, independently of the parity of $n$. Due to space restrictions, we defer the involved proofs of \Cref{lem:DCx} and \Cref{thm:condeven,thm:groupCond} to the appendix.

\subsection{Results for the Condorcet domain}\label{subsec:cond}

In this section, we analyze the set of strategyproof and non-impo\-sing SDSs on the Condorcet domain and its supersets for the case that $n$ is odd. 
In more detail, we will show that, if $n$ is odd, only mixtures of random dictatorships and the Condorcet rule are strategyproof and non-imposing on the Condorcet domain. 
As a byproduct, we also derive a characterization of the Condorcet rule as the only strategyproof, non-imposing, and ``completely non-randomly dictatorial'' SDS on $\mathcal{D}_C$. 
Moreover, we will also prove that, if $n$ is odd, only random dictatorships are strategyproof and non-imposing on every connected superset of the Condorcet domain, thus demonstrating that the Condorcet domain is an inclusion-maximal connected domain that allows for attractive strategyproof SDSs. 

Before proving these claims, we discuss two auxiliary lemmas. First, we show that, if $n$ is odd, every strategyproof and non-im\-po\-sing SDS on a connected super Condorcet domain is also \emph{ex post} efficient. Note that analogous claims are known for a number of domains, e.g., the full domain and the domain of single peaked preferences \citep{Gibb77a,EPS02a}.

\begin{lemma}\label{lem:PO}
     Assume $n\geq 3$ is odd, and let $\mathcal{D}\subseteq \mathcal{R}^N$ denote a connected domain with $\mathcal{D}_C\subseteq \mathcal{D}$. Every strategyproof and non-imposing SDS on $\mathcal{D}$ is \emph{ex post} efficient.
\end{lemma}
\begin{proof}
    Assume $n\geq 3$ is odd and let $\mathcal{D}$ denote a connected domain with $\mathcal{D}_C\subseteq \mathcal{D}$. Moreover, consider a strategyproof and non-imposing SDS $f$ on $\mathcal{D}$ and assume for contradiction that $f$ fails \emph{ex post} efficiency. This means that there are a profile $R^1\in\mathcal{D}$ and two alternatives $x,y\in A$ such that $x \succ_i y$ for all voters $i\in N$ but $f(R^1,y)>0$. Now, consider the profile $R^2$ derived from $R^1$ by making $x$ into the best alternative of every voter $i\in N$. Clearly, $R^2\in\mathcal{D}_C\subseteq \mathcal{D}$ because $x$ is the Condorcet winner in $R^2$. Since $U(\succ_i^1, y)=U(\succ_i^2,y)$ for all voters $i\in N$, there is by connectedness an ad-path from $R^1$ to $R^2$ along which $y$ is never swapped. Hence, we infer that $f(R^2,y)=f(R^1,y)>0$ and $f(R^2,x)<1$ by repeatedly applying localizedness along this ad-path.
    
    Next, let $R^3\in\mathcal{D}$ denote a profile such that $f(R^3,x)=1$; such a profile exists by non-imposition. If $x$ is the Condorcet winner in $R^3$, we can reinforce this alternative until it is top-ranked by every voter without leaving the domain $\mathcal{D}$. This leads to a profile $R^4$ in which $x$ is unanimously top-ranked, and non-perversity shows that $f(R^4,x)\geq f(R^3,x)=1$. Finally, we can again use the connectedness of $\mathcal{D}$ to find an ad-path from $R^4$ to $R^2$ along which $x$ is never swapped. Hence, localizedness requires that $f(R^2,x)=f(R^4,x)= 1$, which contradicts our previous observation. 
    
    As second case, suppose that $x$ is not the Condorcet winner in $R^3$. Since $n$ is odd, there is an alternative $z\in A\setminus \{x\}$ and a set of voters $I$ with $|I|>\frac{n}{2}$ such that $z\succ_i^3 x$ for all $i\in I$. Now, consider the profile $R^5$ derived from $R^3$ by making $z$ into the best alternative of the voters $i\in I$. Clearly, $R^5\in\mathcal{D}_C\subseteq\mathcal{D}$ because more than half of the voters top-rank $z$. Moreover, it holds that $U(\succ_i^3,x)=U(\succ_i^5,x)$ for all $i\in N$, and thus connectedness and localizedness imply that $f(R^5,x)=1$. Next, let $R^6$ denote the profile derived from $R^5$ by making $x$ into the best alternative of the voters $i\in N\setminus I$ and into the second best one of the voters $i\in I$. We can transform $R^5$ into $R^6$ by repeatedly reinforcing $x$, and $z$ stays always the Condorcet winner as it is top-ranked by the voters $i\in I$. Hence, $R^6\in\mathcal{D}$ and non-perversity shows that $f(R^6,x)=1$. Finally, we let the voters $i\in I$ swap $x$ and $z$ one after another. Since all voters top-ranks $x$ or $z$ in $R^6$, one of the alternatives is always top-ranked by more than half of the voters during these steps. Hence, we do not leave $\mathcal{D}$. This process terminates in a profile $R^7$ in which all voters top-rank $x$, and non-perversity shows that $f(R^7,x)=1$. This contradicts again that $f(R^2,x)<1$ as there is an ad-path from $R^7$ to $R^2$ along which we do not move $x$. Since we have a contradiction in both cases, $f$ must be \emph{ex post} efficient. 
\end{proof}

\Cref{lem:PO} is helpful for our analysis because \emph{ex post} efficiency---in contrast to non-imposition---is inherited to subdomains. Since an analogous claim also holds for strategyproofness, we next investigate the set of strategyproof and \emph{ex post} efficient SDSs in the domain $\mathcal{D}_C^x$ where alternative $x$ is always the Condorcet winner. 

\begin{restatable}{lemma}{DCx}\label{lem:DCx}
    Fix an alternative $a\in A$ and let $f$ denote a strategyproof and \emph{ex post} efficient SDS on a super Condorcet domain. There is a random dictatorship $d$ and $\gamma\in \mathbb{R}_{\geq 0}$ such that $f(R)=(1-\gamma)\mathit{COND}(R)+\gamma d(R)$ for all $R\in\mathcal{D}_C^a$. 
\end{restatable}
\begin{proofsketch}
Consider an arbitrary super Condorcet domain $\mathcal{D}$, a strategyproof and \emph{ex post} efficient SDS $f$ on $\mathcal{D}$, and fix an alternative $a\in A$. For proving this lemma, we will investigate the behavior of $f$ on several subdomains of $\mathcal{D}_C^a$. In particular, we first fix a set of voters $I\subseteq N$ with $|I|=\lceil \frac{n+1}{2}\rceil$ and a profile $\hat R\in\mathcal{R}^I$ in which all voters $i\in I$ report $a$ as their favorite alternative. Then, we consider the domain $\mathcal{D}_1^{I, \hat R}$ of profiles in which the voters $i\in I$ report $\hat \succ_i$ and the voters $i\in N\setminus I$ report arbitrary preference relations. In particular, we show that $f$ induces an SDS $g_{\hat R}$ on the domain $\mathcal{R}^{N\setminus I}$ that is non-imposing and strategyproof. The  random dictatorship theorem therefore shows that $g_{\hat R}$ is a random dictatorship. By using the relation between $g_{\hat R}$ and $f$, we then derive that there are values $\gamma_C\geq 0$ and $\gamma_i\geq 0$ for $i\in N\setminus I$ such that $f(R)=\gamma_C^\mathit{COND}(R)+\sum_{i\in N\setminus I}\gamma_i d_i(R)$ for all profiles $R\in\mathcal{D}_1^{I, \hat R}$. For proving the lemma from this point on, we repeatedly enlarge the domain $\mathcal{D}_1^{I,\hat R}$ and show that $f$ can always be represented as a mixture of a random dictatorship and the Condorcet rule. For instance, we consider next the domain $\mathcal{D}_2^I$, where the voters in $I$ have to top-rank $a$. Clearly, every profile $R\in\mathcal{D}_2^I$ is in $\mathcal{D}_1^{I,R'}$, where $R'$ is the restriction of $R$ to the voters in $I$. Since $f$ can be represented for every domain $\mathcal{D}_1^{I,\hat R}$ as a mixture of a random dictatorship and the Condorcet rule, we derive an analogous claim for $\mathcal{D}_2^I$ by showing that it is always the same mixture. By further generalizing the domain like this, we eventually derive the lemma. 
\end{proofsketch}

\Cref{lem:DCx} is itself already a rather strong statement as it characterizes the behavior of all strategyproof and \emph{ex post} efficient SDSs $f$ on the domains $\mathcal{D}_C^x$. In particular, this result does neither require that $n$ is odd nor a connectedness condition on the domain. On the other hand, \Cref{lem:DCx} does not relate the behavior of $f$ for different subdomains $\mathcal{D}_C^x$, and it might be that the weight on the Condorcet rule is negative. Indeed, if $n$ is even and $m=3$, it can be checked that the SDS $f(R)=\sum_{i\in N}\frac{1}{n-1}d_i(R)-\frac{1}{n-1}\mathit{COND}(R)$ is well-defined, non-imposing, and strategyproof for the Condorcet domain because the Condorcet winner is top-ranked by at least one voter if $m=3$. 

Nevertheless, \Cref{lem:DCx} is the central tool for proving all of our theorems and we will use it next to characterize the set of strategyproof and non-imposing SDSs on the Condorcet domain and all of its connected supersets for the case that $n$ is odd.

\begin{restatable}{theorem}{cond}\label{thm:cond}
	Assume $n\geq 3$ is odd and let $\mathcal{D}\subseteq \mathcal{R}^N$ denote a connected domain. The following claims are true.
	\begin{enumerate}[label=(\arabic*), leftmargin=*,topsep=4pt]
	    \item Assume $\mathcal{D}_C=\mathcal{D}$. An SDS on $\mathcal{D}$ is strategyproof and non-imposing if and only if it is a mixture of a random dictatorship and the Condorcet rule. 
	    \item Assume $\mathcal{D}_C\subsetneq \mathcal{D}$. An SDS on $\mathcal{D}$ is strategyproof and non-imposing if and only if it is a random dictatorship.
	\end{enumerate}
\end{restatable}
\begin{proof} 
Assume $n\geq 3$ is odd and let $\mathcal{D}$ denote a connected domain with $\mathcal{D}_C\subseteq \mathcal{D}$.\smallskip

\textbf{Proof of Claim (1)}: First, we assume that $\mathcal{D}=\mathcal{D}_C$ and consider an SDS $f$ on $\mathcal{D}$ that is a mixture of a random dictatorship and the Condorcet rule. Since mixtures of strategyproof SDSs are themselves strategyproof and the Condorcet rule as well as all random dictatorships are known to satisfy this axiom on $\mathcal{D}_C$, it follows immediately that $f$ is strategyproof. Moreover, all random dictatorships and the Condorcet rule choose an alternative with probability $1$ if it is unanimously top-ranked. Since all these profile are in $\mathcal{D}_C$, we derive that $f$ is also non-imposing.
	
	For the other direction, assume that $f$ is a strategyproof and non-imposing SDS on $\mathcal{D}_C$. Since the Condorcet domain is connected if $n$ is odd (see \Cref{lem:connected}), we derive from \Cref{lem:PO} that $f$ is \emph{ex post} efficient. In turn, \Cref{lem:DCx} shows that for every alternative $x\in A$, there are values $\gamma_C^x$ and $\gamma_i^x\geq 0$ for all $i\in N$ such that $f(R)=\gamma_C^x \mathit{COND}(R)+\sum_{i\in N} \gamma_i^x d_i(R)$ for all $R\in\mathcal{D}_C^x$. Hence, the theorem follows by showing that $\gamma_C^x=\gamma_C^y$ and $\gamma_i^x=\gamma_i^y$ for all $x,y\in A$ and all $i\in N$, and that $\gamma_C^x\geq 0$. First, we show that $\gamma_i^x=\gamma_i^y$ and $\gamma_C^x=\gamma_C^y$ for all $x,y\in A$. For doing so, consider three alternatives $a,b,c\in A$, a voter $i\in N$, and the profiles $R^1$ and $R^2$.
	
	\begin{profile}{L{0.1\columnwidth} L{0.3\columnwidth} L{0.3\columnwidth}}
		$R^1$: & $i$: $c,a,b,\dots$ & $N\setminus \{i\}$: $a,b,c,\dots$\\
		$R^2$: & $i$: $c,a,b,\dots$ & $N\setminus \{i\}$: $b,a,c,\dots$
	\end{profile}
	
	Clearly, $R^1\in \mathcal{D}_C^a$ and $R^2\in\mathcal{D}_C^b$ and thus, $f(R^1,c)=\gamma_i^a$ and $f(R^2,c)=\gamma_i^b$. Furthermore, since $\mathcal{D}_C$ is connected and $U(\succ^1_i,c)=U(\succ^2_i,c)$ for all $i\in N$, there is an ad-path from $R^1$ to $R^2$ along which $c$ is never swapped. Localizedness implies therefore that $f(R^1,c)=f(R^2,c)$ and hence, $\gamma_i^a=\gamma_i^b$. Because voter $i$ is chosen arbitrarily, this holds for all voters $i\in N$ and we infer that $\gamma_C^a=1-\sum_{i\in N}\gamma_i^a=1-\sum_{i\in N}\gamma_i^b=\gamma_C^b$. This means that $\gamma_C^x=\gamma_C^y$ and $\gamma_i^x=\gamma_i^y$ for all voters $i\in N$ and alternatives $x,y\in A$

    Next, we will show that $\gamma_C^a\geq 0$ for some $a\in A$. For this step, we partition the set of voters in three disjoint subsets  $I_1$, $I_2$, and $I_3$ such that $|I^1|=|I^2|=\frac{n-1}{2}$ and $|I^3|=1$. Now, let $b,c\in A\setminus \{a\}$ denote distinct alternatives and consider the profiles $R^3$ and $R^4$. 
    
    \begin{profile}{L{0.1\columnwidth} L{0.2\columnwidth} L{0.2\columnwidth} L{0.25\columnwidth}}
		$R^3$: & $I^1$: $a,b,\dots$ & $I^2$: $b,a,\dots$ & $I^3$: $c,a,b,\dots$\\
		$R^4$: & $I^1$: $a,b,\dots$ & $I^2$: $b,a,\dots$ & $I^3$: $c,b,a,\dots$
	\end{profile}

    Alternative $a$ is the Condorcet winner in $R^3$ and alternative $b$ in $R^4$. Hence, $f(R^3,a)=\sum_{i\in I^1} \gamma_i^a+\gamma_C^a$ and $f(R^4,a)=\sum_{i\in I^1} \gamma_i^b$. Next, non-perversity shows that $f(R^3,a)\geq f(R^4,a)$. Since $\gamma_i^a=\gamma_i^b$ for all $i\in N$, we thus infer that $\gamma_C^a\geq 0$. Now, by defining $\gamma_C=\gamma_C^a$ and $\gamma_i=\gamma_i^a$ for all $i\in N$ and some $a\in A$, we conclude that $f(R)=\gamma_C\mathit{COND}(R)+\sum_{i\in N} \gamma_i d_i(R)$ for all $R\in\mathcal{D}_C$.\smallskip
    
    \textbf{Proof of Claim (2):} For the second claim, we assume that $\mathcal{D}_C\subsetneq \mathcal{D}$. Since it is straightforward to see that random dictatorships are strategyproof and non-imposing on $\mathcal{D}$, we focus on the converse. For this, let $f$ denote a strategyproof and non-imposing SDS on $\mathcal{D}$. By \Cref{lem:PO}, $f$ is also \emph{ex post} efficient. As a consequence, it is non-imposing in the Condorcet domain, and we thus infer from Claim (1) that there are $\gamma_C\geq 0$ and $\gamma_i\geq 0$ for all $i\in N$ such that $f(R)=\gamma_C\mathit{COND}(R)+\sum_{i\in N} \gamma_i d_i(R)$ for all $R\in\mathcal{D}_C$. Now, for proving Claim (2), consider a profile  $R\in\mathcal{D}\setminus\mathcal{D}_C$ and let $x$ denote an arbitrary alternative. Since $n$ is odd and there is no Condorcet winner in $R$, there is a set of voters $I$ and an alternative $y\in A\setminus \{x\}$ such that $|I|>\frac{n}{2}$ and $y\succ_i x$ for all $i\in I$. Next, we consider the profile $R'$ derived from $R$ by letting all voters $i\in I$ make $y$ into their favorite alternative. Clearly, $y$ is the Condorcet winner in $R'$ and thus $f(R',x)=\sum_{i\in N} \gamma_i d_i(R',x)$. On the other hand, $U(\succ_i, x)=U(\succ_i',x)$ for all $i\in N$ because the voters in $I$ prefer $y$ to $x$ in $R$. Hence, we can apply connectedness and localizedness to derive that $f(R,x)=f(R',x)=\sum_{i\in N}\gamma_i d_i(R,x)$. Since $x$ is chosen arbitrarily, this means that $f(R)=\sum_{i\in N}\gamma_i d_i(R)$. In particular, it must hold that $\sum_{i\in N}\gamma_i=1$ and thus $\gamma_C=0$ as otherwise $\sum_{x\in A} f(R,x)<1$. This proves $f(R)=\sum_{i\in N} \gamma_i d_i(R)$ for all $R\in\mathcal{D}$.
\end{proof}

Claim (1) of \Cref{thm:cond} immediately implies that the Condorcet rule is the only ``completely non-randomly dictatorial'' SDS on the Condorcet domain that satisfies strategyproofness and non-imposition. To formalize this observation, we introduce the notion of $\gamma$-randomly dictatorial SDSs first suggested by \citet{BLR21b}: a strategyproof SDS $f$ on a domain $\mathcal{D}$ is $\gamma$-randomly dictatorial if $\gamma\in [0,1]$ is the maximal value such that $f$ can be represented as $f(R)=\gamma d(R)+(1-\gamma) g(R)$ for all profiles $R\in\mathcal{D}$, where $d$ is a random dictatorship and $g$ is another strategyproof SDS on $\mathcal{D}$. It follows immediately from \Cref{thm:cond} that, if $n$ is odd, the Condorcet rule is the only $0$-randomly dictatorial, strategyproof, and non-imposing SDS on the Condorcet domain. This corollary generalizes Theorem 1 of \citet{CaKe03a} who have characterized the Condorcet rule with equivalent axioms in the deterministic setting. Furthermore, this insight highlights the appeal of the Condorcet rule on the Condorcet domain because every other strategyproof and non-imposing SDS is a mixture of the Condorcet rule and a random dictatorship. 

On the other hand, Claim (2) of \Cref{thm:cond} generalizes the random dictatorship theorem from the full domain to all connected supersets of $\mathcal{D}_C$ if $n$ is odd. Since deterministic voting rules can be seen as a special case of SDSs, our result also generalizes the Gibbard-Satterhwaite theorem to these smaller domains. 
In particular, Claim (2) of \Cref{thm:cond} shows that adding even a single profile to the Condorcet domain can turn the positive results of Claim (1) into a negative one. This follows, for instance, by considering the subsequent domain $\mathcal{D}_1=\mathcal{D}_C\cup \{R^*\}$. The preference profile $R^*$ is shown below, where $I=\{4,6,\dots, n-1\}$, $J=\{5,7,\dots, n\}$.\smallskip

\begin{profile}{L{0.1\columnwidth} L{0.22\columnwidth} L{0.22\columnwidth} L{0.22\columnwidth}}
$R^*$: & $1$: $a,b,c,\dots$ & $2$: $b,c,a,\dots$ & $3$: $c,a,b,\dots$\\
& $I$: $a,b,c,\dots$ & $J$: $c,b,a,\dots$
\end{profile}

\subsection{Results for the tie-breaking Condorcet domain}\label{subsec:condeven}

A natural follow-up question to \Cref{thm:cond} is to ask for the strategyproof and non-imposing SDSs on the Condorcet domain if $n$ is even. Unfortunately, since the Condorcet domain is not connected in this case, a concise characterization of all these SDSs seems impossible. We therefore characterize the set of strategyproof and non-imposing SDS on the tie-breaking Condorcet domain $\mathcal{D}_C^\ssucc$. Moreover, the following theorem also demonstrates that, if $n$ is even, tie-breaking Condorcet domains are inclusion-maximal connected domains that allow for strategyproof and non-imposing SDSs other than random dictatorships.

\begin{restatable}{theorem}{condeven}\label{thm:condeven}
Assume $n\geq 4$ is even, let ${\ssucc}\in\mathcal{R}$ be a preference relation, and $\mathcal{D}\subseteq \mathcal{R}^N$ be a connected domain. The following claims hold.
\begin{enumerate}[label=(\arabic*), leftmargin=*,topsep=4pt]
        \item Assume $\mathcal{D}=\mathcal{D}_C^\ssucc$. An SDS on $\mathcal{D}$ is strategyproof and non-impo\-sing if and only if it is a mixture of a random dictatorship and the tie-breaking Condorcet rule $\mathit{COND}^\ssucc$. 
        \item Assume $\mathcal{D}_C^\ssucc\subsetneq \mathcal{D}$. An SDS on $\mathcal{D}$ is strategyproof and non-impo\-sing if and only if it is a random dictatorship.
    \end{enumerate}
\end{restatable}
\begin{proofsketch}
Assume $n\geq 4$ is even, fix a preference relation ${\ssucc}\in\mathcal{R}$, and consider a connected domain $\mathcal{D}$ with $\mathcal{D}_C^\ssucc\subseteq \mathcal{D}$. First, note that random dictatorships are strategyproof and non-imposing on $\mathcal{D}$, regardless of whether $\mathcal{D}=\mathcal{D}_C^\ssucc$ or $\mathcal{D}_C^\ssucc\subsetneq\mathcal{D}$. Moreover, if $\mathcal{D}=\mathcal{D}_\ssucc$, then $\mathit{COND}^\ssucc$ is strategyproof on $\mathcal{D}$ because every manipulation of this rule can be turned in a manipulation of the Condorcet rule for $n+1$ voters. Since mixture of strategyproof SDSs are strategyproof, it follows that all mixtures of random dictatorships and $\mathit{COND}^\ssucc$ are strategyproof, and it is easy to see that these rules are also non-imposing.

Next, we focus on the direction from left to right and consider for this a strategyproof and non-imposing SDS $f$ on $\mathcal{D}$. Analogous to \Cref{lem:PO}, it is not difficult to derive that $f$ is \emph{ex post} efficient on $\mathcal{D}$. Hence, \Cref{lem:DCx} implies that there are values $\gamma_i^x\geq 0$ for all $i\in N$ and $\gamma_C^x$ such that $f(R)=\gamma_C^x\mathit{COND}(R)+\sum_{i\in N}\gamma_i^xd_i(R)$ for all subdomains $\mathcal{D}_C^x$ and profiles $R\in\mathcal{D}_C^x$. Next, we show analogously to the proof of \Cref{thm:cond} that $\gamma_C^x=\gamma_C^y$ and $\gamma_i^x=\gamma_i^y$ for all $x,y\in A$ and $i\in N$ and we can thus drop the superscript. Since $\mathit{COND}(R)=\mathit{COND}^\ssucc(R)$ for all $R\in\mathcal{D_C}$ if $n$ is even, this means that $f(R)=\gamma_C\mathit{COND}^\ssucc(R)+\sum_{i\in N} \gamma_i d_i(R)$ for all $R\in\mathcal{D}_C$. 

Now, to prove the first claim, let us assume that $\mathcal{D}=\mathcal{D}_C^\ssucc$. In this case, we first show that $f(R)$ can also be represented as $\gamma_C\mathit{COND}^\ssucc(R)+\sum_{i\in N} \gamma_i d_i(R)$ if there is an alternative $x$ in $R$ that is top-ranked by at least half of the voters and the Condorcet winner in $(R,\ssucc)$. Next, we consider a profile $R\in\mathcal{D}_C^\ssucc$ and let $x$ denote the Condorcet winner in $(R,\ssucc)$. This means that for every alternative $y\in A\setminus \{x\}$, there are at least $\frac{n}{2}$ voters who prefer $x$ to $y$ in $R$. If we let these voters reinforce $x$ until it is top-ranked, we arrive at a profile $R'$ such that $f(R')=\gamma_C\mathit{COND}^\ssucc(R')+\sum_{i\in N} \gamma_i d_i(R')$. Moreover, connectedness and localizedness imply that the probability of $y$ does not change when going from $R$ to $R'$. Since $y\in A\setminus \{x\}$ is chosen arbitrarily, we derive from this observation that $f(R)=\gamma_C\mathit{COND}^\ssucc(R)+\sum_{i\in N} \gamma_i d_i(R)$ for every profile $R\in\mathcal{D}_C^\ssucc$. As last step, we show that $\gamma_C\geq 0$ by using a similar argument as in the proof of \Cref{thm:cond}. This completes the proof of Claim (1).

For proving Claim (2), assume that $\mathcal{D}_C^\ssucc\subsetneq \mathcal{D}$. Using the same reasoning as for Claim (1), we infer that $f$ can be represented as a mixture of a random dictatorship and $\mathit{COND}^\ssucc$ for all profiles $R\in\mathcal{D}_C^\ssucc$. Now, consider a profile $R\in\mathcal{D}\setminus \mathcal{D}_C^\ssucc$. For the proof, we identify for every alternative $x\in A$ a profile $R'\in\mathcal{D}_C^\ssucc$ such that $U(\succ_i,x)=U(\succ_i',x)$ for all voters $i\in N$. Once we have these profiles, the proof proceeds exactly as the proof of Claim (2) in \Cref{thm:cond}.
\end{proofsketch}

First, note that \Cref{thm:condeven} implies---analogously to \Cref{thm:cond}---that the tie-breaking Condorcet rule is the only strategyproof, non-imposing, and $0$-randomly dictatorial SDS on the tie-breaking Condorcet domain if $n$ is even. In particular, this proves again that choosing the Condorcet winners is desirable because $\mathit{COND}^\ssucc$ chooses the Condorcet winners whenever there is one. Moreover, since the tie-breaking Condorcet domain is only a small extension of the Condorcet domain, this result demonstrates the important role of Condorcet winners for the existence of strategyproof and non-imposing SDSs other than random dictatorships. 

Furthermore, Claim (2) in \Cref{thm:condeven} shows again that adding even a single profile to $\mathcal{D}_C^\ssucc$ can turn the positive result into a negative one. In particular, note that this claim also implies that the domain of all profiles with a weak Condorcet winner only allows for random dictatorships as strategyproof and non-imposing SDSs.

\begin{remark}	
An important observation of \Cref{thm:cond,thm:condeven} is that every strategyproof and non-imposing SDS on the respective domains can be represented as a mixture of deterministic voting rules, each of which is strategyproof and non-imposing. This is sometimes called deterministic extreme point property and it is remarkable that many important domains satisfy this condition \citep[][]{RoSa20a}. On one side, this shows that randomization does not lead to completely new strategyproof SDSs. On the other hand, the deterministic extreme point property allows for a natural interpretation of strategyproof and non-imposing SDSs: we decide by chance which deterministic voting rule is executed.
\end{remark}
\vspace{-0.3cm}
	\noindent\begin{table*}
	\setlength{\tabcolsep}{4pt}
		\begin{tabular}{L{0.22\textwidth} L{0.2\textwidth} L{0.22\textwidth} L{0.29\textwidth}}
			\toprule
			& Full domain $\mathcal{R}^N$ & Domains $\mathcal{D}$ with $\mathcal{D}_C^{(\ssucc)}\subsetneq \mathcal{D}$ & (tie-breaking) Condorcet domain $\mathcal{D}_C^{(\ssucc)}$\\\midrule 
			Deterministic, strategyproof, and non-imposing voting rules & Dictatorships \citep{Gibb73a,Satt75a} & \emph{Dictatorships$^\diamond$ (\Cref{thm:cond,thm:condeven})} & Dictatorships and the (tie-breaking) Condorcet rule (\Cref{thm:condeven} and \citep{CaKe03a})\\\addlinespace
			Strategyproof and non-imposing SDSs & Random dictatorships \citep{Gibb77a} & \emph{Random dictatorships$^\diamond$ (\Cref{thm:cond,thm:condeven})} & \emph{Mixtures of random dictatorships and the (tie-breaking) Condorcet rule (\Cref{thm:cond,thm:condeven})}\\\addlinespace
			Group-strategyproof and non-imposing SDSs & Dictatorial SDSs \citep{Barb79a} & \emph{Dictatorial SDSs (\Cref{thm:groupCond})} & \emph{Dictatorial SDSs and the (tie-breaking) Condorcet rule (\Cref{thm:groupCond})}\\
			\bottomrule
		\end{tabular}
		\caption{Comparison of results for the full domain $\mathcal{R}^N$, strict supersets of $\mathcal{D}_C$ (resp. $\mathcal{D}_C^\ssucc$), and the (tie-breaking) Condorcet domain $\mathcal{D}_C$ (resp. $\mathcal{D}_C^\ssucc$). Each row characterizes a set of SDSs for the full domain $\mathcal{R}^N$, strict supersets of $\mathcal{D}_C$ (resp. $\mathcal{D}_C^\ssucc$), and the (tie-breaking) Condorcet domain $\mathcal{D}_C$ (resp. $\mathcal{D}_C^\ssucc$), respectively. For the last two columns, the results rely on a case distinction with respect to $n$: if $n$ is odd, we consider the results of \Cref{thm:cond} for the Condorcet domain and its supersets; if $n$ is even, we consider the results of \Cref{thm:condeven} for the tie-breaking Condorcet domain and its supersets. The results marked with a diamond ($\diamond$) require that the considered domain is connected. New results are italicized.}
		\label{tab:FDvsCD}
	\end{table*}
	
\begin{remark} The connectedness condition is required for Claim (2) in \Cref{thm:cond,thm:condeven} because there are domains $\mathcal{D}$ with $\mathcal{D}_C\subsetneq \mathcal{D}$ (resp. $\mathcal{D}_C^\ssucc\subsetneq \mathcal{D}$) that allow for non-imposing and strategyproof SDSs that are no random dictatorships. For example, consider the domain $\mathcal{D}_2$ which is derived by adding a single preference profile $R^1$ to the Condorcet domain. If $R^1$ differs from every profile in $\mathcal{D}_C$ in the preference relations of at least two voters, an arbitrary outcome can be returned for $R^1$ without violating strategyproofness.
\end{remark}

\subsection{Results based on Group-strategyproofness}\label{subsec:GSP}
	
Finally, we investigate the set of of group-strategyproof and non-imposing SDSs on the Condorcet domain and its supersets. In particular, we will show that only the Condorcet rule and dictatorial SDSs satisfy group-strategyproofness on the Condorcet domain. Note that this result is independent of the parity of $n$ and group-strategyproofness thus allows for a unified characterization. Moreover, we also prove a counterpart to Claim (2) in \Cref{thm:cond,thm:condeven}, which notably does not require connectedness. 

\begin{restatable}{theorem}{groupCond}\label{thm:groupCond}
Assume $n\geq 3$ and let $\mathcal{D}\subseteq \mathcal{R}^N$ denote an arbitrary domain. The following claims are true.
\begin{enumerate}[label=(\arabic*), leftmargin=*,topsep=4pt]
    \item Assume $\mathcal{D}=\mathcal{D}_C$. An SDS on $\mathcal{D}$ is group-strategyproof and non-imposing if and only if it is a dictatorship or the Condorcet rule.
    \item Assume $\mathcal{D}_C\subsetneq \mathcal{D}$ and that there is a profile $R\in\mathcal{D}$ such that for each $x\in A$, there is $y\in A$ with $g_R(y,x)>0$. An SDS on $\mathcal{D}$ is group-stra\-te\-gy\-proof and non-imposing if and only if it is a dictatorship. 
\end{enumerate}
\end{restatable}
\begin{proofsketch}
For the direction from right to left of both claims, we note first that dictatorships are clearly non-imposing and group-strategyproof on every super Condorcet domain. Furthermore, it is also apparent that the Condorcet rule is non-imposing on the Condorcet domain. We hence only need to show that $\mathit{COND}$ is group-strategyproof on $\mathcal{D}_C$. For this, let $I\subseteq N$ denote a non-empty set of voters and consider two profiles $R, R'\in \mathcal{D}_C$ such that ${\succ}_i={\succ'}_i$ for all $i\in N\setminus I$. Moreover, let $c$ and $c'$ denote the respective Condorcet winners in $R$ and $R'$. If $c=c'$, then $\mathit{COND}(R)=\mathit{COND}(R')$ and the Condorcet rule is clearly group-strategyproof. On the other hand, if $c\neq c'$, there must be a voter $i\in I$ with $c \succ_i c'$ and $c' \succ_i' c$; otherwise, it is impossible that $g_{R}(c,c')>0$ and $g_{R'}(c',c)>0$. However, this voter prefers $\mathit{COND}(R)$ to $\mathit{COND}(R')$, which proves that $\mathit{COND}$ is also in this case group-strateygproof.

For the other direction, we consider a group-strategyproof and non-imposing SDS $f$ on a domain $\mathcal{D}$ with $\mathcal{D}_C\subseteq\mathcal{D}$. First, it is not difficult to see that $f$ must be \emph{ex post} efficient. Since group-strategyproofness implies strategyproofness, we can now invoke \Cref{lem:DCx} to derive that for every alternative $x\in A$, there are values $\gamma_C^x$ and $\gamma_i^x\geq 0$ for all $i\in N$ such that $f(R)=\gamma_C^x\mathit{COND}(R)+\sum_{i\in N} \gamma_i^x d_i(R)$ for all $R\in\mathcal{D}_C^x$. Moreover, we can essentially use the same argument as in the proof of \Cref{thm:cond} to show that $\gamma_C^x=\gamma_C^y$ and $\gamma_i^x=\gamma_i^y$ for all $i\in N$ and $x,y\in A$. We hence drop the superscript from now on and write, e.g., $\gamma_C$ instead of $\gamma_C^x$. 

Next, we show that $\gamma_i=1$ if $\gamma_i>0$. For this, we assume for contradiction there is a voter $i\in N$ with $0<\gamma_i\neq 1$. We consider the profiles $R^1$ and $R^2$ shown below to derive a contradiction.

\begin{profile}{L{0.1\columnwidth} L{0.3\columnwidth} L{0.3\columnwidth}}
		$R^1$: & $i$: $c,a,b,\dots$ & $N\setminus \{i\}$: $b,a,c,\dots$\\
		$R^2$: & $i$: $a,b,c,\dots$ & $N\setminus \{i\}$: $a,b,c,\dots$
	\end{profile}

Since $b$ is the Condorcet winner in $R^1$ and $\gamma_i\neq 1$, we have that $f(R^1,c)=\gamma_i>0$ and $f(R^1,b)=1-f(R^3,c)>0$. On the other hand, \emph{ex post} efficiency shows that $f(R^2,a)=1$. However, the set of all voters can now group-manipulate by deviating from $R^1$ to $R^2$ because $f(R^1,U(\succ^1_j,a))<1=f(R^2,U(\succ^1_j,a))$ for all $j\in N$. This contradicts that $f$ is group-strategyproof and thus proves that $\gamma_i=1$ if $\gamma_i>0$. Now, since there clearly cannot be different voters $i,j$ with $\gamma_i=1$ and $\gamma_j=1$, we infer that for all profiles $R\in\mathcal{D}_C$, either $f(R)=d_i(R)$ for some $i\in N$ or $f(R)=\mathit{COND}(R)$ if $\gamma_i=0$ for all $i\in N$. This proves Claim (1) by choosing $\mathcal{D}=\mathcal{D}_C$.

For proving Claim (2), we assume next that there is a profile $R^*\in\mathcal{D}$ such that for every alternative $x\in A$, there is another alternative $y\in A\setminus\{x\}$ such that $g_{R^*}(y,x)>0$. Now, consider an alternative $a\in A$ with $f(R^*,a)>0$, let $b$ denote an alternative with $g_{R^*}(b,a)>0$, and define $I=\{i\in N\colon b \succ_i^* a\}$. We let all voters $i\in I$ make $b$ into their best alternative to derive the profile $R'$. Note that $R'\in\mathcal{D}_C\subseteq \mathcal{D}$ as $y$ is the Condorcet winner in $R'$. If $f(R')=\mathit{COND}(R')$, the voters $i\in I$ can group-manipulate by deviating from $R$ to $R'$ because they all prefer $b$ to $a$. Hence, group-strategyproofness requires that there is a voter $i\in N$ such that $f(R)=d_i(R)$ for all $R\in\mathcal{D}_C$. 
From here on, it is easy to see that $f=d_i(R)$ for all $R\in\mathcal{D}$, which proves Claim (2).
\end{proofsketch}

\Cref{thm:groupCond} generalizes \Cref{thm:cond} to super Condorcet domains for an even number of voters by using group-strategy\-proofness. In particular, it entails that the Condorcet rule is the only group-strategy\-proof, non-imposing, and non-dictatorial SDS on the Condorcet domain, regardless of the parity of $n$. Moreover, Claim (2) of the theorem shows that the Condorcet domain is essentially a maximal domain that allows for a group-strategyproof and non-imposing SDS apart from dictatorships. In more detail,if $n$ is odd, every domain $\mathcal{D}$ with $\mathcal{D}_C\subsetneq\mathcal{D}$ satisfies the conditions of Claim (2) in \Cref{thm:groupCond}. Hence, \emph{no} superset of the Condorcet domain admits group-strategyproof and non-imposing SDSs other than dictatorships if $n$ is odd. On the other hand, if $n$ is even, \Cref{thm:groupCond} can be refined. For instance, $\mathit{COND}^\ssucc$ is also group-strategyproof on $\mathcal{D}_C^\ssucc$. Indeed, it is possible to prove an exact equivalent of \Cref{thm:condeven} for disconnected domains based on group-strategyproofness. 

\begin{remark} The results of \citet{Barb79a} imply that every group-strategyproof and non-imposing SDS on the full domain is a dictatorship. Hence, \Cref{thm:groupCond} and \citeauthor{Barb79a}'s results share a common idea: group-strategyproof and non-imposing SDSs cannot rely on randomization to determine the winner. However, whereas only undesirable SDSs are group-strategyproof and non-imposing on $\mathcal{R}^N$, the attractive Condorcet rule satisfies these axioms on $\mathcal{D}_C$.
\end{remark}

\section{Conclusion}

We study strategyproof and non-imposing social decision schemes (SDSs) on the Condorcet domain (which consists of all preference profiles with a Condorcet winner) and its supersets. These domains are of great relevance because empirical results suggest that real-world election commonly admit a Condorcet winner. In contrast to the full domain, there are attractive strategyproof SDSs on the Condorcet domain: we show that, if the number of voters $n$ is odd, every strategyproof and non-imposing SDS on the Condorcet domain can be represented as a mixture of a random dictatorship and the Condorcet rule. An immediate consequence of this insight is that the Condorcet rule is the only strategyproof, non-imposing, and completely non-randomly dictatorial SDS on the Condorcet domain if $n$ is odd. Moreover, we demonstrate that, if $n$ is odd, the Condorcet domain is a maximal connected domain that allows for strategyproof and non-imposing SDSs other than random dictatorships. We also derive analogous results for even $n$ by slightly extending the Condorcet domain. Finally, we investigate  the set of group-strategyproof and non-imposing SDSs on super Condorcet domains: we prove that the Condorcet rule is the only non-dictatorial, group-strategyproof, and non-imposing SDS on the Condorcet domain, and that no SDS satisfies these axioms on larger domains. 

Our results for the Condorcet domain show an astonishing similarity to classic results for the full domain, but have a more positive flavor. For instance, while the random dictatorship theorem shows that only mixtures of dictatorial SDSs are strategyproof and non-imposing on the full domain, we prove in \Cref{thm:cond} that mixtures of dictatorial SDSs and the Condorcet rule are the only strategyproof and non-imposing SDSs on the Condorcet domain (if the number of voters is odd). A more exhaustive comparison between results for the full domain and for the Condorcet domain is given in \Cref{tab:FDvsCD}. In particular, our results highlights the important role of the Condorcet rule on the Condorcet domain: even if we allow for randomization, it is still the most appealing strategyproof voting rule. Thus, our theorems make a strong case for choosing the Condorcet winner whenever it exists. 

\section*{Acknowledgements}

{This work was supported by the Deutsche Forschungsgemeinschaft under grant \mbox{BR 2312/12-1}.
We thank the anonymous reviewers for helpful comments.
}

\newpage
\appendix

\section{Omitted Proofs}

In this appendix, we discuss the omitted proofs of \Cref{lem:connected,lem:DCx} and \Cref{thm:condeven,thm:groupCond}. We start by formally proving \Cref{lem:connected}.

\connected*
\begin{proof}
The proofs for the Condorcet domain $\mathcal{D}_C$ and odd $n\geq 3$ and for the tie-breaking Condorcet domain $\mathcal{D}_C^\ssucc$ and even $n\geq 4$ work almost identically. In particular, we can use the same ad-paths for both domains and only the arguments on why the ad-paths remain in the respective domain slightly change. We therefore focus only on the Condorcet domain in this proof. 

Hence, assume that $n\geq 3$ is odd, consider two profiles $R, R'\in\mathcal{D}_C$, and let $c$ and $c$' denote the respective Condorcet winners. We proceed in multiple steps to show that $\mathcal{D}_C$ is connected. First, we show that this domain is weakly connected by constructing an ad-path from $R$ to $R'$. Next, we suppose that there is an alternative $x\in A$ such that $U(\succ_i, x)= U(\succ_i',x)$ for all $i\in N$. We then consider several cases depending on whether $c=c'$ and $x\in \{c,c'\}$ and always construct an ad-path from $R$ to $R'$ along which $x$ is not moved.\medskip

\textbf{Step 1: $\mathcal{D}_C$ is weakly connected}

For proving this claim, we need to construct an ad-path from $R$ to $R'$ in $\mathcal{D}_C$. For doing so, we start at $R$ by reinforcing $c$ until it is the best alternative of all voters $i\in N$. This leads to a profile $R^1$ and, since we only reinforce the Condorcet winner, we do not leave the Condorcet domain. Next, if $c\neq c'$, we reinforce $c'$ until we arrive in the profile $R^2$ in which every voter ranks $c$ first and $c'$ second. During all these steps, $c$ is unanimously top-ranked and therefore the Condorcet winner. Next, we let the voters swap $c$ and $c'$ one after another. Since $n$ is odd, $c$ or $c'$ are top-ranked by more than half of the voters in each of these profiles, which shows that we do not leave the Condorcet domain. Finally, we are now in a profile $R^3$ in which all voters top-rank $c'$. Note that if $c=c'$, we can simply set $R^3=R^1$. As the next step, we use swaps to reorder the alternatives in $A\setminus \{c'\}$ according to $R'$. These steps result in the profile $R^4$ which only differs from $R'$ in the fact that $c'$ is unanimously top-ranked. Since this also holds during all steps in the construction of this profile, we do not leave the Condorcet domain. Finally, we weaken $c'$ to derive $R'$ from $R^4$. Note that, since $U(\hat \succ_i,c')\subseteq U(\succ_i',c')$ for all voters $i\in N$ and intermediate profiles $\hat R$, $ c'$ is the Condorcet winner during all steps. Hence, there is an ad-path from $R$ to $R'$ in $\mathcal{D}_C$ and the Condorcet domain is weakly connected if $n$ is odd.\medskip

\textbf{Step 2.1: $x=c$}

Next, we suppose that there is an alternative $x\in A$ such that $U(\succ_i,x)= U(\succ_i',x)$ for all $i\in N$. In this case, we need to construct an ad-path from $R$ to $R'$ along which $x$ is not moved. For constructing this ad-path, we use a case distinction with respect to $x$, $c$, and $c'$, and first assume that $x=c$. This implies that $c=c'$ because we cannot change the Condorcet winner without moving this alternative. As a consequence, we can simply reorder the alternatives $y,z\in A\setminus \{c\}$ to transform $R$ to $R'$. Since $U(\succ_i, x)=U(\succ_i',x)$ for all $i\in N$, we never need to swap $x$ and thus also do not leave the Condorcet domain.\medskip

\textbf{Step 2.2: $c=c'$ and $x\neq c$}

As the second case, suppose that $c=c'$ but $x\neq c$ and let $I=\{i\in N\colon c \succ_i x\}$ denote the set of voters who prefer $c$ to $x$ in $R$. As the first step, we consider the profile $R^1$ which is defined as follows: the voters $i\in I$ top-rank $c$ and order all alternatives as in $R$, and the voters $i\in N\setminus I$ rank $c$ directly below $x$ and order the alternatives $A\setminus \{c\}$ according to $R$. Clearly, we can move from $R$ to $R^1$ by only reinforcing $c$ and thus, $R^1$ and all intermediate profiles are in the Condorcet domain. As the second step, we analyze the profile $R^2$ which is defined as follows: all voters $i\in I$ top-rank $c$ and order the remaining alternatives according to $R'$, and all voters $i\in N\setminus I$ place $c$ directly below $x$ and reorder the remaining alternatives according to $R'$. Since $U(\succ_i',x)=U(\succ_i,x)=U(\succ_i^2,x)$ for all $i\in N$, no voter needs to swap $x$ during any of these steps. Indeed, $x$ partitions the alternatives in two sets for every voter $i$: $U_i=\{y\in A\colon y\succ_i x\}$ and $L_i=\{y\in A\colon x\succ_i y\}$. Since we only need to reorder alternatives within these sets to go from $R^1$ to $R^2$, no swap involves $x$. Moreover, since $c$ is the best alternative within $U_i$ (if $i\in I$) or $L_i$ (if $i\in N\setminus I$) in both $R^1$ and $R^2$, we also do not need to swap $c$. This ensures that we do not leave the Condorcet domain during these steps. Finally, we can go from $R^2$ to $R'$ by only weakening $c$. Since $U(\hat \succ_i,c)\subseteq U(\succ_i',c)$ for all $i\in N$ and profiles $\hat R$ on this ad-path, $c$ always remains the Condorcet winner. This completes the ad-path from $R$ to $R'$ and it is easy to see that $x$ is never swapped with another alternative along it.\medskip

\textbf{Step 2.3: $c\neq c'$ and $x\not \in \{c,c'\}$}

As last case, we assume that $c\neq c'$ and $x\not\in \{c,c'\}$. In this case, let $I_1=\{i\in N\colon c \succ_i x\}$ and $I_2=\{i\in N\colon c' \succ_i x\}$ denote the sets of voters who prefer $c$ and $c'$, respectively, to $x$. Since $U(\succ_i,x)=U(\succ_i',x)$ for all $i\in N$, it does not matter whether we define $I_1$ and $I_2$ with respect to $R$ or $R'$. Now, consider the profile $R^1$ constructed as follows: 
\begin{itemize}[leftmargin=*,topsep=4pt]
    \item All $i\in N$ order all alternatives in $A\setminus \{c,c'\}$ as in $R'$.
    \item All $i\in I_1\setminus I_2$ prefer $c$ the most and rank $c'$ directly below $x$.
    \item All $i\in I_2\setminus I_1$  prefer $c'$ the most and rank $c$ directly below $x$.
    \item All $i\in I_1\cap I_2$ prefer $c$ the most and $c'$ the second most. 
    \item All $i\in N\setminus (I_1\cup I_2)$ rank $c$ directly below $x$ and $c'$ directly below~$c$.
\end{itemize}
    
    In particular, note that $c$ is top-ranked by all voters in $I_1$ and placed directly below $x$ by all voters $i\in N\setminus I_1$. Moreover, it holds that $U(\succ_i^1,x)=U(\succ_i,x)$ for all $i\in N$. This implies that $U(\succ_i^1,c)\subseteq U(\succ_i,c)$ for all $i\in N$ and thus, $c$ is the Condorcet winner in $R^1$. We can therefore use the construction of Step 2.2 to find an ad-path from $R$ to $R^1$ along which $x$ is never moved. 
    
    Next, let $I=(I_1\cap I_2)\cup (N\setminus(I_1\cup I_2))$ and consider the profile $R^2$ derived from $R^1$ by letting the voters $i\in I$ swap $c$ and $c'$ one after another. Note that at $I\neq \emptyset$ because $y \succ_i x \succ_i z'$ implies that $y\succ_i' x \succ_i' z$ for all voters $i\in N$. Or put differently, the voters in $I_1\setminus I_2$ and $I_2\setminus I_1$ cannot swap $c$ and $c'$ as this would require them to change the upper contour set of $x$. Since $c\neq c'$, it thus follows that $I\neq \emptyset$. Next, note that $c \succ_i y$ implies that $c\succ_i^k y$ and $c' \succ_i' y$ implies $c'\succ_i^k y$ for all voters $i\in N$, alternatives $y\in A\setminus \{c,c'\}$, and $k\in \{1,2\}$. This claim holds since we are not allowed to swap $x$ with any alternative, and $c$ and $c'$ are either top-ranked or directly below $x$. In particular, this means that $g_{R^k}(c,y)>0$ and $g_{R^k}(c',y)>0$ for all $y\in A\setminus \{c,c'\}$ and $k\in \{1,2\}$. Finally, this analysis also holds for all profiles $\hat R$ on the ad-path between $R^1$ and $R^2$ and thus $c$ is the Condorcet winner if $g_{\hat R}(c,c')>0$ and $c'$ otherwise. Moreover, in $R^2$, $c'$ must be the Condorcet winner because all voters $i\in I^2$ top-rank $c'$ and all other voters place it directly below $x$, i.e., $U(\succ_i^2,c')\subseteq U(\succ_i',c')$ for all $i\in N$. Finally, we can use again the construction of Step 2.2 to go from $R^2$ to $R'$ as $c'$ is the Condorcet winner in both profiles. This completes the proof.
\end{proof}

Next, we turn to the proof of \Cref{lem:DCx}. 

\DCx*
\begin{proof}
Let $\mathcal{D}$ denote a super Condorcet domain, and let $f:\mathcal{D}\rightarrow\Delta(A)$ denote a strategyproof and \emph{ex post} efficient SDS. Moreover, fix an alternative $a\in A$. For proving the lemma, we will reason about multiple subdomains of $\mathcal{D}_C^a$.

Throughout the proof, we use some additional notation. First, we define $t=\lceil\frac{n+1}{2}\rceil$ as the smallest integer larger than $\frac{n}{2}$. Moreover, given a set of voters $I$ with $\emptyset\subsetneq I\subsetneq N$ and two profiles $R^1\in\mathcal{R}^{I}$, $R^2\in \mathcal{R}^{N\setminus I}$, we define the $R^1+R^2$ as the profile $R$ on the electorate $N$ with ${\succ_i}={\succ^1_i}$ if $i\in I$ and ${\succ_i}={\succ^2_i}$ if $i\in N\setminus I$. Finally, we define the rank of an alternative $x$ in a preference relation $\succ_i$ as $r(\succ_i, x)=|\{y\in A\colon y\succ_i x\lor y=x\}|$. For instance, $r(\succ_i,x)=1$ means that $x$ is voters $i$'s favorite alternative.\medskip

\textbf{Step 1:} Fix a set of voters $I\subseteq N$ with $|I|=t$ and an alternative $b\in A\setminus \{a\}$. Moreover, let $\mathcal{D}_1^{I,b}=\{R\in\mathcal{R}^N\colon \forall i\in I\colon r(\succ_i,a)=1\land \forall i\in N\setminus I\colon r(\succ_i,b)=1\}$ denote the domain in which all voters in $I$ always top-rank $a$ and all voters in $N\setminus I$ always top-rank $b$. We show that $f(R,a)+f(R,b)=1$ for all profiles $R\in\mathcal{D}_1^{I,b}$. For this, consider an arbitrary profile $R\in\mathcal{D}_1^{I,b}$ and let $R^*\in\mathcal{D}_1^{I,b}$ denote the profile shown below.

\begin{profile}{L{0.1\columnwidth} L{0.3\columnwidth} L{0.3\columnwidth}}
        $R^*$: & $I$: $a,b,\dots$ & $N\setminus I$: $b,a,\dots$
\end{profile}

\emph{Ex post} efficiency requires for $R^*$ that $f(R^*,a)+f(R^*,b)=1$ as all other alternatives are Pareto-dominated. Next, consider the profiles $R^1$ and $R^2$ which are defined as follows: in $R^1$, the voters $i\in I$ report $\succ^*_i$ and the voters $i\in N\setminus I$ report $\succ_i$. Conversely, in $R^2$, the voters $i\in I$ report $\succ_i$ and the voters $i\in N\setminus I$ report $\succ^*_i$. Next, we show that $f(R,a)=f(R^1,a)=f(R^*,a)$. Note for this that $U(\succ_i^*,b)=U(\succ^1_i,b)$ for all voters $i\in N$ because ${\succ_i^1}={\succ_i^*}$ for all voters $i\in I$ and $r(\succ_i^*,b)=r(\succ_i^1,b)=1$ for all voters $i\in N\setminus I$. Now, there is clearly an ad-path from $R^*$ to $R^1$ along which $b$ is never swapped and we can therefore infer from localizedness that $f(R^1,b)=f(R^*,b)$. Moreover, $b$ Pareto-dominates all other alternatives $z\in A\setminus \{a,b\}$ in $R^1$, so \emph{ex post} efficiency requires that $f(R^1,a)+f(R^1,b)=1$. Hence, we conclude that $f(R^1,a)=1-f(R^1,b)=1-f(R^*,b)=f(R^*,a)$. 
	
As the next step, we transform $R^1$ into $R$ by reordering the alternatives in $A\setminus \{a\}$ in the preference relations of the voters $i\in I$. This time, we have that $U(\succ^1_i,a)=U(\succ_i,a)$ for all voters $i\in N$ because $r(\succ_i^1,a)=r(\succ_i, a)=1$ for all $i\in I$ and ${\succ_i^1}={\succ_i}$ for all $i\in N\setminus I$. Hence, an analogous argument as before shows that $f(R,a)=f(R^1,a)=f(R^*,a)$. Moreover, we can use a symmetric argument to derive that $f(R,b)=f(R^2,b)=f(R^*,b)$. This proves that $f(R,a)+f(R,b)=1$ for all preference profiles $R\in \mathcal{D}_1^{I,b}$.\medskip

\textbf{Step 2:} Fix a set of voters $I\subseteq N$ with $|I|=t$ and a profile $R^*\in\mathcal{R}^I$ in which all voters top-rank $a$. In this step, we consider the domain $\mathcal{D}_2^{I,R^*}=\{R^*+R\colon R\in\mathcal{R}^{N\setminus I}\}$, i.e., the voters $i\in I$ have to report $\succ^*_i$ and the voters in $N\setminus I$ can report arbitrary preference relations. Our goal is to show that there are values $\gamma_C\geq 0$ and $\gamma_i\geq 0$ for all $i\in N\setminus I$ such that $f(R)=\gamma_C \mathit{COND}(R)+\sum_{i\in N\setminus I} \gamma_i d_i(R)$ for all $R\in\mathcal{D}_2^{I,R^*}$.

For proving this, let $\bar R$ denote a profile in $\mathcal{D}_2^{I,R^*}$ in which all voters $i\in N\setminus I$ prefer $a$ the least and define $\delta_a=f(\bar R,a)$. We first show that $f(R,a)\geq \delta_a$ for all profiles $R\in\mathcal{D}_2^{I,R^*}$. For this, consider an arbitrary profile $R\in\mathcal{D}_2^{I,R^*}$ and let $R'$ denote the profile derived from $R$ by letting each voter $i\in N\setminus I$ make $a$ into his least preferred alternative. Clearly, there is an ad-path from $\bar R$ to $R'$ in $\mathcal{D}_2^{I,R^*}$ along which $a$ is never swapped and thus, $f(R',a)= f(\bar R, a)=\delta_a$ because of localizedness. On the other hand, we only need to reinforce $a$ to go from $R'$ to $R$. Hence, non-perversity shows that $f(R,a)\geq f(R',a)=\delta_a$, which proves our claim.

Now, if $\delta_a=1$, this means that $f(R,a)=1$ and therefore $f(R)=\mathit{COND}(R)$ for all profiles $R\in\mathcal{D}_2^{I,R^*}$. In this case, Step 2 is proven as we can choose $\gamma_C=1$ and $\gamma_i=0$ for $i\in N\setminus I$. Thus, we suppose that $0\leq \delta_a<1$ and define the SDS $g_{R^*}(R)$ for the electorate $N\setminus I$ as follows: $g_{R^*}(R,x)=\frac{1}{1-\delta_a} f(R^*+R,x)$ for $x\in A\setminus \{a\}$ and $g_{R^*}(R,a)=\frac{1}{1-\delta_a}(f(R^*+R,a)-\delta_a)$. In particular, note that $g_{R^*}$ is defined on $\mathcal{R}^{N\setminus I}$, i.e., it is defined on the full domain with respect to $N\setminus I$. Subsequently, we show that $g_{R^*}$ is a well-defined, strategyproof, and non-imposing SDS because the random dictatorship theorem then entails that $g_{R^*}$ is a random dictatorship.

First, we prove that $g_{R^*}$ is well-defined and note for this that $\sum_{x\in A} g_{R^*}(R,x)=\frac{1}{1-\delta_a} \sum_{x\in A} f(R^*+ R,x) - \frac{\delta_a}{1-\delta_a}=1$ for every profile $R\in \mathcal{R}^{N\setminus I}$ because $\sum_{x\in A} f(R^*+ R,x)=1$. Moreover, it clearly holds that $g_{R^*}(R,x)=\frac{1}{1-\delta_a} f(\bar R+R,x)\geq 0$ for every $x\in A\setminus \{a\}$. Finally, $g_{R^*}(R,a)=\frac{1}{1-\delta_a} (f(R^*+R,a)-\delta_a)\geq 0$ because $f(R^*+R,a)\geq \delta_a$. Hence, $g_{R^*}$ is indeed a well-defined SDS. Moreover, $g_{R^*}$ inherits strategyproofness from $f$. In more detail, if a voter $i\in N\setminus I$ could manipulate $g_{R^*}$ by deviating from $R$ to $R'$, then he could also manipulate $f$ by deviating from $R^*+R$ to $R^*+R'$. Finally, we note that $g_{R^*}$ is non-imposing. For $a$, this follows immediately because $g_{R^*}(R)=\frac{1}{1-\delta_a}(f(R^*+R,a)-\delta_a)=1$ for every profile $R\in\mathcal{R}^{N\setminus I}$ in which all voters $i\in N\setminus I$ top-rank $a$ because \emph{ex post} efficiency requires that $f(R^*+R,a)=1$. For the alternatives $x\in A\setminus \{a\}$, this follows by considering a profile $R\in\mathcal{R}^{N\setminus I}$ such that all voters top-rank $x$ and bottom-rank $a$. It is not difficult to see that $f(R^*+R,a)=f(\bar R,a)=\delta_a$ because of localizedness. On the other hand, Step 1 shows that $f(R^*+R,a)+f(R^*+R,x)=1$. Hence, we conclude that $g_{R^*}(R,x)=\frac{1}{1-\delta_a} f(R^*+R,x)=1$ and $g_{R^*}$ is indeed non-imposing. 

Since $g_{R^*}$ is a strategyproof and non-imposing SDS on the full domain (with respect to $N\setminus I$), the random dictatorship theorem shows that $g_{R^*}$ must be a random dictatorship. Or, put differently, there are values $\gamma_i\geq 0$ for $i\in N\setminus I$ such that $g_{R^*}(R)=\sum_{i\in N\setminus I} \gamma_i d_i(R)$ for all $R\in\mathcal{R}^{N\setminus I}$. Since $g_{R^*}(R,x)=\frac{1}{1-\delta_a} f(R^*+R,x)$ for $x\in A\setminus \{a\}$ and $g_{R^*}(R,a)=\frac{1}{1-\delta_a} (f(R^*+R,a)-\delta_a)$, it follows that $f(R^*+R,x)=(1-\delta_a)\sum_{i\in N\setminus I} \gamma_i d_i(R^*+R,x)$ and $f(R^*+R,a)=\delta_a+(1-\delta_a)\sum_{i\in N\setminus I}\gamma_i d_i(R^*+R,a)$. Finally, let $\gamma_C'=\delta_a$ and $\gamma_i'=(1-\delta_a)\gamma_i$ for $i\in N\setminus I$. Since $a$ is the Condorcet winner for all $R\in\mathcal{D}_2^{I,R^*}$, it is easy to see that $f(R)=\gamma'_C\mathit{COND}(R)+\sum_{i\in N\setminus I} \gamma_i' d_i(R)$ for all $R\in\mathcal{D}_2^{I, R^*}$.\medskip

\textbf{Step 3:} Once again, we fix a set of voters $I\subseteq N$ with $|I|=t$. In this step, we analyze the domain $\mathcal{D}_3^I=\{R\in\mathcal{R}^N\colon \forall i\in I\colon r(\succ_i,a)=1\}$, i.e., the voters in $I$ have to top-rank $a$, but otherwise the domain is not constrained. Our goal is to show that there are $\gamma_C\geq 0$ and $\gamma_i\geq 0$ for $i\in N\setminus I$ such that $f(R)=\gamma_C\mathit{COND}(R)+\sum_{i\in N\setminus I} \gamma_i d_i(R)$. 

For this, consider two profiles $R^1, R^2\in\mathcal{R}^I$ in which all voters $i\in I$ unanimously top-rank $a$. By Step 3, there are $\gamma_C^x\geq 0$ and $\gamma_i^x\geq 0$ such that $f(R^x+R)=\gamma_C^x\mathit{COND}(R^x+R)+\sum_{i\in N\setminus I} \gamma_i d_i(R^x+R)$ for all profiles $R\in\mathcal{R}^{N\setminus I}$ and $x\in \{1,2\}$. Our goal is to show that $\gamma_C^1=\gamma_C^2$ and $\gamma_i^1=\gamma_i^2$ for all $i\in N\setminus I$. Assume for contradiction that this is not the case, which implies that $R^1\neq R^2$. Now, consider an ad-path $(\hat R^1, \dots, \hat R^l)$ from $R^1$ to $R^2$ in $\mathcal{R}^I$ along which $a$ is always unanimously top-ranked. Clearly, $\hat R^k+R\in\mathcal{D}_3^I$ for all $k\in \{1,\dots, l\}$ and we can thus also use Step 2 for each of these intermediate profiles. In particular, there must be two consecutive profiles $R^3$ and $R^4$ on this ad-path such that $\gamma_C^3\neq \gamma_C^4$ or $\gamma_i^3\neq \gamma_i^4$ for some $i\in N\setminus I$. Since $R^3$ and $R^4$ are consecutive on the ad-path, $R^4$ evolves out of $R^3$ by swapping two alternatives $x,y\in A\setminus \{a\}$ in the preference relation of a single voter $i^*\in I$.

First suppose that $\gamma_C^3\neq \gamma_C^4$ and consider a profile $R\in\mathcal{R}^{N\setminus I}$ such that $a$ is not top-ranked by any voter. Using Step 2, we infer that $f(R^3+R,a)=\gamma_C^3$ and $f(R^4+R,a)=\gamma_C^4$ because $R^3+R\in \mathcal{D}_2^{I, R^3}$ and $R^4+R\in\mathcal{D}_2^{I, R^4}$. Hence, $f(R^3+R,a)\neq f(R^4+R,a)$. However, this contradicts localizedness as we can transform $R^3+R$ into $R^4+R$ by only swapping $x$ and $y$ in the preference relation of voter $i^*$. Hence, we must have $\gamma_C^3= \gamma_C^4$. As second case, suppose that there is a voter $i\in N\setminus I$ such that $\gamma_i^3\neq \gamma_i^4$. In this case, consider the profile $R\in\mathcal{R}^{N\setminus I}$ such that only voter $i$ top-ranks $a$. Using again the insights of Step 2, we derive that $f(R^3+R,a)=\gamma_C^3+\gamma_i^3\neq \gamma_C^4+\gamma_i^4=f(R^4+R,a)$. In particular, we use here that $\gamma_C^3=\gamma_C^4$ because of the first case. Now, we can again move from $R^3+R$ to $R^4+R$ by only swapping $x$ and $y$ in the preference relation of voter $i^*$, and we thus have a contradiction to localizedness. This shows that $\gamma_C^3=\gamma_C^4$ and $\gamma_i^3=\gamma_i^4$ for all $i\in N\setminus I$ and consequently, the same claim holds for all profiles $R^1, R^2\in\mathcal{R}^I$ in which $a$ is unanimously top-ranked. 

Based on this insight, choose an arbitrary profile $\hat R\in\mathcal{R}^I$ in which all voters top-rank $a$ and let $\gamma_C\geq 0$ and $\gamma_i\geq 0$ for $i\in N\setminus I$ denote the respective values such that $f(\hat R+R)=\gamma_C\mathit{COND}(\hat R+R)+\sum_{i\in N\setminus I} \gamma_i d_i(\hat R+R)$ for all $R\in\mathcal{R}^{N\setminus I}$. It follows from our previous analysis and Step 2 that $f(R)=\gamma_C\mathit{COND}(R)+\sum_{i\in N\setminus I} \gamma_i d_i(R)$ for all $R\in\mathcal{D}_3^I$ because $R\in\mathcal{D}_2^{I,R'}$, where $R'$ is the restriction of $R$ to the voters $i\in I$.\medskip

\textbf{Step 4:} In this step, we analyze $f$ on the domain $\mathcal{D}_4$ which consists of the profiles in which at least $t$ voters top-rank $a$, i.e., $\mathcal{D}_4=\{R\in \mathcal{R}^n\colon \exists I\subseteq N\colon |I|=t\land \forall i\in I\colon r(\succ_i,a)=1\}$. Our goal is to find values $\gamma_C$ and $\gamma_i\geq 0$ for all $i\in N$ such that $f(R)=\gamma_C\mathit{COND}(R) + \sum_{i\in N} \gamma_i d_i(R)$ for all $R\in\mathcal{D}_4$. In particular, note that $\gamma_C$ can be negative after this step.

For this, consider two sets of voters $I_1\subseteq N$ and $I_2\subseteq N$ with $|I_1|=|I_2|=t$. By Step 4, there are values $\gamma^x_C\geq 0$ and $\gamma_i^x\geq 0$ for $i\in N\setminus I_x$ such that $f(R)=\gamma^x_C \mathit{COND}(R)+\sum_{i\in N\setminus I_x} \gamma_i^x d_i(R)$ for all profiles $R\in\mathcal{D}_3^{I_x}$ and $x\in \{1,2\}$. Next, we show that $\gamma_i^{1}=\gamma_i^{2}$ for all voters $i\in N\setminus (I_1\cup I_2)$. For this, consider the profile $R$ in which all voters $j\in N\setminus \{i\}$ top-rank $a$, and voter $i$ top-ranks another alternative $x$. Note that $R$ is both in $\mathcal{D}^{I_1}_3$ and $\mathcal{D}^{I_2}_3$. Hence, it follows from Step 4 that $f(R,x)=\gamma_i^{1}$ and $f(R,x)=\gamma_{i}^{2}$. This implies that $\gamma_i^1=\gamma_i^2$, which proves our claim.

Next, consider three sets of voters $I_1, I_2, I_3$ such that $|I_1|=|I_2|=|I_3|=t$ and $N=(N\setminus I_1)\cup (N\setminus I_2)\cup (N\setminus I_3)$. We define $\gamma_i=\gamma_i^{I_1}$ if $i\in N\setminus I_1$, $\gamma_i=\gamma_i^{I_2}$ if $i\in I_1\setminus I_2$, and $\gamma_i=\gamma_i^{I_3}$ if $i\in (I_1\cap I_2)\setminus I_3$. Here, we use $\gamma_i^{I_x}$ for the values derived for $I_x$ in Step 4. In particular, $\gamma_i$ is defined for all $i\in N$ and $\gamma_i\geq 0$. Moreover, let $\gamma_C=1-\sum_{i\in N}\gamma_i$. We claim that $f(R)=\gamma\mathit{COND}(R)+\sum_{i\in N}\gamma_i d_i(R)$ for all $R\in\mathcal{D}_4$. For showing this, consider an arbitrary profile $R\in \mathcal{D}_{4}$. By the definition of this domain, there is a set of voters $I$ with $|I|=t$ such that all voters in $I$ top-rank $a$. Hence, we can use Step 4 to derive that $f(R)=\gamma^I_C \mathit{COND}(R)+\sum_{i\in N\setminus I} \gamma_i^I d_i(R)$ for some values $\gamma^I_C\geq 0$ and $\gamma_i^I\geq 0$ for all $i\in N\setminus I$. By the insights of the last paragraph, it holds that $\gamma_i^I=\gamma_i^{I_1}$ if $i\in N\setminus I_1$, $\gamma_i^I=\gamma_i^{I_2}$ if $i\in I_1\setminus I_2$, and $\gamma_i^I=\gamma_i^{I_3}$ if $i\in (I_1\cap I_2)\setminus I_3$. This means that $\gamma_i^I=\gamma_i$ for all $i\in N\setminus I$ and that $\gamma^I_C=1-\sum_{i\in N\setminus I} \gamma_i^I=1-\sum_{i\in N\setminus I} \gamma_i=\gamma_C+\sum_{i\in I}\gamma_i$. Since all voters $i\in I$ top-rank $a$, it follows that $f(R)=\gamma_C \mathit{COND}(R) + \sum_{i\in N}\gamma_i d_i(R)$, which proves this step.\medskip

\textbf{Step 5:} Finally, we consider the domain $\mathcal{D}_C^a$ that contains all profiles in which $a$ is the Condorcet winner. Note that $\mathcal{D}_4\subseteq \mathcal{D}_C^a$. Hence, there are values $\gamma_C$ and $\gamma_i\geq 0$ for all $i\in N$ such that $f(R)=\gamma_C\mathit{COND}(R)+\sum_{i\in N} \gamma_i d_i(R)$ for all $R\in\mathcal{D}_4$ because of Step 4. We will show that the same holds for all profiles $R\in\mathcal{D}_C^a$.

For this, consider an arbitrary profile $R\in\mathcal{D}_C^a$. If $R\in\mathcal{D}_4$, the claim follows immediately and we hence suppose that $R\not\in\mathcal{D}_4$. Now, consider an alternative $x\in A\setminus \{a\}$. We claim that $f(R,x)=\sum_{i\in N} \gamma_i d_i(R,x)$. For proving this, note that there is a set of voters $I\subseteq N$ such that $|I|=t$ and $a \succ_i x$ for all $i\in I$ because $a$ is the Condorcet winner in $R$. Next, consider the profile $R'$ derived from $R$ by letting the voters $i\in I$ reinforce $a$ until it is their most preferred alternative. Clearly, $R'\in \mathcal{D}_4$ and we infer that $f(R',x)=\sum_{i\in N} \gamma_i d_i(R',x)$. On the other hand, we do not move $x$ in the transition from $R$ to $R'$ as all voters in $I$ already prefer $a$ to $x$. Hence, $d_i(R,x)=d_i(R',x)$ for all voters $i\in N$ and localizedness shows that $f(R,x)=f(R',x)=\sum_{i\in N} \gamma_i d_i(R,x)$. Finally, since this holds for all $x\in A\setminus \{a\}$, it follows that $f(R,a)=1-\sum_{x\in A\setminus \{a\}} f(R,x)=1-\sum_{i\in N\colon r(\succ_i,a)>1} \gamma_i=\gamma_C+\sum_{i\in N\colon r(\succ_i,a)=1} \gamma_i=\gamma_C\mathit{COND}(R,a)+\sum_{i\in N} \gamma_i d_i(R,a)$. In particular, we use here that $\gamma_C=1-\sum_{i\in N}\gamma_i$ by the definition in Step 5. Since $R$ is chosen arbitrarily, this proves that $f(R)=\gamma_C\mathit{COND}(R)+\sum_{i\in N} \gamma_i d_i(R)$ for all $R\in\mathcal{D}_C^a$.
\end{proof}

Next, we turn to the proof of \Cref{thm:condeven}.

\condeven*
\begin{proof}
    Assume $n\geq 4$ is even, fix a preference relation ${\ssucc}\in\mathcal{R}$, and consider an arbitrary connected domain $\mathcal{D}$ with $\mathcal{D}_C^\ssucc\subseteq \mathcal{D}$. First, we show that every strategyproof and non-imposing SDS on $\mathcal{D}$ is \emph{ex post} efficient. For this, we assume for contradiction that $f$ is such an SDS but fails \emph{ex post} efficiency, which means that there are a profile $R^1\in \mathcal{D}$ and two alternatives $x,y\in A$ such that $x\succ_i^1 y$ for all $i\in N$ but $f(R^1,y)>0$. By connectedness and localizedness, it follows that $f(R^2,y)>0$ for the profile $R^2$ derived from $R^1$ by making $x$ into the favorite alternative of every voter. On the other hand, there is by non-imposition a profile $R^3$ such that $f(R^3,x)=1$. If $x$ is the Condorcet winner in $(R^3, \ssucc)$, we can reinforce $x$ until it is top-ranked by all voters and derive a contradiction just as in \Cref{lem:PO}. 
    
    Hence, suppose that there is an alternative $z\in A$ such that $g_{(R^3,\ssucc)}(z,x)>0$. Moreover, if there are multiple such alternatives, we assume that $z$ is the best alternative according to $\ssucc$, i.e., $z\ssucc z'$ if $g_{(R^3,\ssucc)}(z,x)>0$ and $g_{(R^3,\ssucc)}(z',x)>0$. Now, let $I=\{i\in N\colon z\succ_i^3 x\}$ denote the set of voters who prefer $z$ to $x$ in $R^3$ and note that $|I|\geq\frac{n}{2}$. We consider the profile $R^4$ derived from $R^3$ by letting the voters $i\in I$ make $z$ into their favorite alternative and the voters $i\in N\setminus I$ ranks $z$ directly below $x$. Our next goal is to show that $R^4\in\mathcal{D}_C^\ssucc\subseteq\mathcal{D}$ because $z$ is the Condorcet winner in $(R^4,\ssucc)$. If $|I|>\frac{n}{2}$, this is clear and we hence suppose that $|I|=\frac{n}{2}$. This means that $z\ssucc x$ as otherwise $g_{(R^3,\ssucc)}(z,x)<0$. If $z$ was not the Condorcet winner in $(R^4, \ssucc)$, there is another alternative $z'\in A\setminus \{x\}$ such that $g_{(R^4,\ssucc)}(z',z)>0$. This is only possible if all voters $i\in N\setminus I$ prefer $z'$ to $z$ and if $z'\ssucc z$. Because all voters $i\in N\setminus I$ rank $z$ directly below $x$ in $R^4$, $z'\succ_i^4 z$ implies $z'\succ_i^4 x$ for all $i\in N\setminus I$ and the transitivity of $\ssucc$ shows that $z'\ssucc x$. However, then $g_{(R^3,\ssucc)}(z',x)>0$ and $z'\ssucc z$, which contradicts the definition of $z$. Hence, no such alternative $z'$ exists and $z$ is indeed the Condorcet winner in $(R^4,\ssucc)$. Moreover, note that $U(\succ_i^3,x)=U(\succ_i^4,x)$ for all $i\in N$. Localizedness and connectedness thus show that $f(R^3,x)=f(R^4,x)$. From this observation on, we can derive a contradiction analogous to the proof of \Cref{lem:PO}. Hence, it follows that every strategyproof and non-imposing SDS on $\mathcal{D}$ is indeed \emph{ex post} efficient.
    
    Next, we will prove Claims~(1) and~(2).
    
    \smallskip
    
    \textbf{Proof of Claim (1)}: Suppose that $\mathcal{D}=\mathcal{D}_C^\ssucc$. We first prove the direction from right to left and show that the tie-breaking Condorcet $\mathit{COND}^\ssucc$ rule is strategyproof and non-imposing on $\mathcal{D}_C^\ssucc$. For this, note that for every $R\in\mathcal{D}_C^\ssucc$, the profile $(R,\ssucc)$ is in the Condorcet domain for $n+1$ voters. Even more, $\mathit{COND}^\ssucc(R)=\mathit{COND}(R,\ssucc)$ for all these profiles. Since $\mathit{COND}$ is strategyproof on the Condorcet domain, this implies that $\mathit{COND}^\ssucc$ is strategyproof on the tie-breaking Condorcet domain $\mathcal{D}_C^\ssucc$. Moreover, $\mathit{COND}^\ssucc$ is clearly non-imposing as it chooses an alternative with probability $1$ if it is unanimously top-ranked. It follows now from the same arguments as in the proof of Claim (1) of \Cref{thm:cond} that every mixture of a random dictatorship and the tie-breaking Condorcet rule is strategyproof and non-imposing on $\mathcal{D}_C^\ssucc$.
    
    For the other direction, let $f$ denote a strategyproof and non-imposing SDS on the tie-breaking Condorcet domain $\mathcal{D}_C^\ssucc$. By our previous arguments, it follows that $f$ is \emph{ex post} efficient. Hence, \Cref{lem:DCx} applies for $f$ and shows that for every alternative $x\in A$, there are values $\gamma_C^x$ and $\gamma_i^x\geq 0$ for all $i\in N$ such that $f(R')=\gamma_C^x\mathit{COND}(R')+\sum_{i\in N} \gamma_i^xd_i(R')$ for all $R'\in\mathcal{D}_C^x$. From this observation on, we proceed in multiple steps to prove Claim (1). In particular, we show first that $\gamma^x_C=\gamma^y_C$ and $\gamma_i^x=\gamma_i^y$ for all $i\in I$ and $x,y\in A$. This means that $f(R)=\gamma_C\mathit{COND}^\ssucc(R')+\sum_{i\in N} \gamma_id_i(R)$ for all profiles $R\in\mathcal{D}_C$, where $\gamma_C=\gamma_C^x$ and $\gamma_i=\gamma_i^x$ for all $i\in N$ and some alternative $x\in A$. As our next step, we fix an alternative $a\in A$ and consider the domain $\mathcal{D}_C^{\ssucc,a}$ that precisely consists of the profiles $R$ such that $a$ is the Condorcet winner in $(R,\ssucc)$. For this domain, we prove in three steps that $f(R)=\gamma_C\mathit{COND}^\ssucc(R)+\sum_{i\in N} \gamma_id_i(R)$ for all profiles $R\in\mathcal{D}_C^{\ssucc,a}$. Finally, we show that $\gamma_C\geq0$, which completes the proof of the theorem since $a$ is chosen arbitrarily.\medskip 
    
    \textbf{Step 1:} Consider two alternatives $a,b\in A$. We will show that $\gamma^a_C=\gamma^b_C$ and $\gamma_i^a=\gamma_i^b\geq 0$ for all voter $i\in I$. For this, let $i$ denote an arbitrary voter, $c$ an alternative in $A\setminus \{a,b\}$, and consider the profiles $R^1$ and $R^2$.
    
    \begin{profile}{L{0.06\columnwidth} L{0.3\columnwidth} L{0.3\columnwidth}}
    $R^1$: & $i: c,a,b\dots$ & $N\setminus \{i\}$: $a,b,c\dots$\\
    $R^2$: & $i: c,a,b\dots$ & $N\setminus \{i\}$: $b,a,c\dots$
    \end{profile}
    
    Clearly, $a$ is the Condorcet winner in $R^1$ and $b$ in $R^2$. Hence, $f(R^1,c)=\gamma_i^a$ and $f(R^2,c)=\gamma_i^b$. Moreover, since $\mathcal{D}_C^\ssucc$ is connected, there is an ad-path from $R^1$ to $R^2$ that does not move $c$. Thus, localizedness shows that $\gamma_i^a=\gamma_i^b$. Since this holds for all voters, we also infer that $\gamma_C^a=1-\sum_{i\in N} \gamma_i^a=1-\sum_{i\in N} \gamma_i^b=\gamma_C^b$. Now, let $\gamma_C=\gamma_C^a$ and $\gamma_i=\gamma_i^a$ for all voters $i\in N$ and some $a
    \in A$. It is easy to see that $f(R)=\gamma_C\mathit{COND}^\ssucc(R)+\sum_{i\in N} \gamma_i d_i(R)$ for all $R\in\mathcal{D}_C$. In particular, we use here that $\mathit{COND}^\ssucc(R)=\mathit{COND}(R)$ for all profiles $R\in\mathcal{D}_C$ if $n$ is even.\medskip
    
    \textbf{Step 2:} Fix a set of voters $I\subseteq N$ with $|I|=\frac{n}{2}$ and an alternative $a\in A$. In this step, we consider the domain $\mathcal{D}_1^{I,a}$ which consists of all profiles $R$ such that $a$ is the Condorcet winner in $(R,\ssucc)$ and top-ranked by all voters $i\in I$. Our goal is to show that $f(R,a)=\gamma_C\mathit{COND}^\ssucc(R,a)+\sum_{i\in N} \gamma_id_i(R,a)$ for all profiles $R\in\mathcal{D}_1^{I,a}$. 
    
    First, note that this claim follows immediately from Step 1 for all profiles $R\in\mathcal{D}_C\cap D_1^{I,a}$. Hence, we focus on profiles $R\in\mathcal{D}_1^{I,a}\setminus \mathcal{D}_C$. Since $a$ is by definition the Condorcet winner in $(R,\ssucc)$ for all these profiles $R$, there is an alternative $b\in A\setminus \{a\}$ such that $b \succ_i a$ for all voters $i\in N\setminus I$ and $a \ssucc b$. In particular, this means that no voter in $N\setminus I$ top-ranks $a$, and we thus need to show that $f(R,a)=\gamma_C+\sum_{i\in I} \gamma_i$ for all such profiles $R$. 
    
    For doing so, consider a profile $R^1\in\mathcal{D}_1^{I,a}\setminus \mathcal{D}_C$ and let $b$ denote an alternative such that $b \succ_i^1 a$ for all $i\in N\setminus I$ and $a\ssucc b$. Furthermore, let $R^2$ denote the profile derived from $R^1$ by letting all voters $i\in N\setminus I$ make $b$ into their favorite alternative, and all voters $i\in I$ make $b$ into their second best alternative (after $a$). Since $a$ is the Condorcet winner in $(R^1,\ssucc)$ and $U(\succ_i^1,a)=U(\succ_i^2,a)$ for all $i\in N$, it is also the Condorcet winner in $(R^2,\ssucc)$. Moreover, we can use connectedness and localizedness to derive that $f(R^1,a)=f(R^2,a)$ since we did not move $a$ for transforming $R^1$ to $R^2$. Next, observe that $f(R^2,x)=0$ for all $x \in A\setminus \{a,b\}$ as $b$ Pareto-dominates all those alternatives. Finally, consider the profile $R^3$, which we derive from $R^2$ by letting all voter $i\in N\setminus I$ make $a$ into their second best alternative. Connectedness implies that there is an ad-path from $R^2$ to $R^3$ along which $b$ is never swapped and thus $f(R^2,b)=f(R^3,b)$. On the other hand, \emph{ex post} efficiency still requires that $f(R^3,x)=0$ for $x\in A\setminus\{a,b\}$ and thus, $f(R^2,a)=1-f(R^2,b)=1-f(R^3,b)=f(R^3,a)$. 
    
    Now, to prove that $f(R^1,a)=f(R^3,a)=\gamma_C+\sum_{i\in I}\gamma_i$, we consider the following profiles, where $c$ is an arbitrary alternative in $A\setminus\{a,b\}$ and $i^*$ is a voter in $N\setminus I$. Note that all these profiles are in $\mathcal{D}_C^\ssucc$: for $R\in \{R^4,R^5,R^6\}$, it is easy to see that $a$ is the Condorcet winner in $(R,\ssucc)$ because $a\ssucc b$ and $n\geq 4$. For $R^7$, $c$ is the Condorcet winner in $(R^7,\ssucc)$ because it is top-ranked by all voters but one.
    
    \begin{profile}{L{0.06\columnwidth} L{0.2\columnwidth} L{0.35\columnwidth} L{0.2\columnwidth}}
        $R^4$: & $I$: $a,b,c,\dots$ & $N\setminus (I\cup \{i^*\})$: $b,c,a,\dots$ & $i^*$: $b,a,c,\dots$\\
        $R^5$: & $I$: $a,c,b,\dots$ & $N\setminus (I\cup \{i^*\})$: $c,b,a,\dots$ & $i^*$: $b,a,c,\dots$\\
        $R^6$: & $I$: $a,c,b,\dots$ & $N\setminus (I\cup \{i^*\})$: $c,b,a,\dots$ & $i^*$: $a,b,c,\dots$\\
        $R^7$: & $I$: $c,a,b,\dots$ & $N\setminus (I\cup \{i^*\})$: $c,b,a,\dots$ & $i^*$: $b,a,c,\dots$
    \end{profile}
    
    In $R^4$, all alternatives but $a$ and $b$ are still Pareto-dominated and thus have probability $0$. Moreover, $U(\succ^3_i,b)=U(\succ^4_i,b)$ for all $i\in N$, which implies that $f(R^4,b)=f(R^3,b)$ because $\mathcal{D}_C^\ssucc$ is connected and $f$ is localized. In turn, we derive that $f(R^4,a)=1-f(R^4,b)=1-f(R^3,b)=f(R^3,a)$. Next, we can again use localizedness and connectedness to infer that $f(R^5,a)=f(R^4,a)$. 
    
    As last step, note that $a$ is the Condorcet winner in $R^6$ and $c$ is the Condorcet winner in $R^7$. Hence, for $x\in \{6,7\}$, $R^x\in \mathcal{D}_C$ and $f(R^x)=\gamma_C\mathit{COND}^\ssucc(R^x)+\sum_{i\in N} \gamma_i d_i(R^x)$ by Step 1. This observation means that $f(R^6,c)=\sum_{i\in N\setminus (I\cup \{i^*\})} \gamma_i$ and $f(R^7,b)=\gamma_{i^*}$. Next, note that $U(\succ_i^5,c)=U(\succ^6_i,c)$ and $U(\succ_i^5,b)=U(\succ_i^7,b)$ for all $i\in N$. Hence, there is an ad-path from $R^5$ to $R^6$ (resp. $R^5$ to $R^7$) along which $c$ (resp. $b$) is never swapped because $\mathcal{D}_C^\ssucc$ is connected. We infer now from localizedness that $f(R^5,c)=f(R^6,c)=\sum_{i\in N\setminus (I\cup \{i^*\})} \gamma_i$ and $f(R^5,b)=f(R^7,c)=\gamma_{i^*}$. Moreover, \emph{ex post} efficiency requires that $f(R^5,x)=0$ for all $x\in A\setminus \{a,b,c\}$. We can therefore deduce that $f(R^5,a)=1-f(R^5,\{b,c\})=1-\sum_{i\in N\setminus I} \gamma_i=\gamma_C + \sum_{i\in I} \gamma_i$. Since $f(R^1,a)=f(R^3,a)=f(R^5,a)$, this proves Step~2.\medskip
    
    \textbf{Step 3:} For the third step, we fix a set of voters $I\subseteq N$ with $|I|=\frac{n}{2}$ and an alternative $a\in A$. We will show that $f(R)=\gamma_C\mathit{COND}^\ssucc(R)+\sum_{i\in N} \gamma_id_i(R)$ for all profiles $R\in\mathcal{D}_1^{I,a}$.
    
    For this, consider a profile $R^1\in\mathcal{D}_1^{I,a}$. If there is a voter $i\in N\setminus I$ who top-ranks $a$ in $R^1$, then $R^1\in\mathcal{D}_C$ and the claim follows from Step 1. On the other hand, if there is no such voter, consider the profile $R^2$ derived from $R^1$ by making $a$ into the second best alternative of all voters $i\in N\setminus I$. Since none of these voters top-rank $a$ in $R^1$, it holds that $U(\succ^2_i,a)\subseteq U(\succ^1_i, a)$ for all $i\in N$. In particular, we can transform $R^1$ into $R^2$ by only reinforcing $a$ against a single alternative in every step. Hence, there is an ad-path $(\hat R^1, \dots, \hat R^l)$ from $R^1$ to $R^2$ along which $a$ is only reinforced. In particular, it is easy to see that $a$ is the Condorcet winner in $(\hat R^k,\ssucc)$ for every profile $\hat R^k$ on this ad-path. We therefore do not leave the domain $\mathcal{D}_1^{I,a}$ during this process. Moreover, it follows from this observation that $f(\hat R^k,a)=\gamma_C+\sum_{i\in I}\gamma_i$ for all $k\in \{1,\dots, l\}$ because of the insights of Step 2. Next, localizedness shows that $f(\hat R^k,z)=f(\hat R^{k-1},z)$ for all alternatives $z\in A$ except the alternative $x$ that is weakened against $a$. However, since $f(\hat R^k,a)=f(\hat R^{k-1},a)$ the probability of $x$ cannot change either. This proves that $f(\hat R^k)=f(\hat R^{k-1})$ for all $k\in \{2,\dots,l\}$ and therefore that $f(R^1)=f(R^2)$. 
    
    Now, if $a$ is the Condorcet winner in $R^2$, then $f(R^1)=f(R^2)=\gamma_C\mathit{COND}^\ssucc(R^2)+\sum_{i\in N} \gamma_id_i(R^2)=\gamma_C\mathit{COND}^R(R^1) +\sum_{i\in N} \gamma_id_i(R^1)$, where the last equality follows because no voter changes his favorite alternative. On the other hand, if $R^2\not\in\mathcal{D}_C$, then all voters in $N\setminus I$ top-rank another alternative $b$. In particular, this means that all alternatives $x\in A\setminus \{a,b\}$ are Pareto-dominated by $a$ and thus, $f(R^1,b)=f(R^2,b)=1-f(R^2,a)=1-(\gamma_C+\sum_{i\in I}\gamma_i)=\sum_{i\in N\setminus I}\gamma_i$. Hence, it holds again that $f(R^1)=\gamma_C\mathit{COND}^\ssucc(R^1)+\sum_{i\in N} \gamma_id_i(R^1)$, which completes the proof of this step.\medskip
    
    \textbf{Step 4:} Next, we show that $f(R)=\gamma_C\mathit{COND}^\ssucc(R)+\sum_{i\in N} \gamma_id_i(R)$ for all profiles $R\in\mathcal{D}_C^{\ssucc,a}$. For this, consider an arbitrary profile $R$ in this domain and an alternative $x\in A\setminus \{a\}$. Since $a$ is the Condorcet winner in $(R,\ssucc)$, there are at least $\frac{n}{2}$ voters in $R$ who prefer $a$ to $x$. Now, consider the profile $R'$ derived from $R$ by letting these voters make $a$ into their best alternative. By the connectedness of $\mathcal{D}_C^\ssucc$, there is an ad-path from $R$ to $R'$ along which $x$ is never moved. Hence, localizedness implies that $f(R,x)=f(R',x)$. On the other hand, $R'\in\mathcal{D}_1^{I,a}$ for some set $I$ and thus, $f(R',x)=\sum_{i\in N}\gamma_i d_i(R',x)$. Moreover, $x$ was not top-ranked by any of the voters who reinforced $a$ because $a \succ_i' x$. Therefore, $f(R,x)=\sum_{i\in N}\gamma_i d_i(R,x)$. Since $x$ is chosen arbitrarily in $A\setminus\{a\}$, we infer that $f(R,x)=\sum_{i\in N}\gamma_i d_i(R,x)$ for all $x\in A\setminus \{a\}$ and $R\in\mathcal{D}_C^{\ssucc,a}$. This implies that $f(R,a)=1-\sum_{x\in A\setminus \{a\}} f(R,x)=1-\sum_{i\in N\colon r(\succ_i,a)>1} \gamma_i=\gamma_C+\sum_{i \in N\colon r(\succ_i,a)=1} \gamma_i$. We thus derive that $f(R)=\gamma_C\mathit{COND}^\ssucc(R)+\sum_{i\in N} \gamma_id_i(R)$ for all profiles $R\in\mathcal{D}_C^{\ssucc,a}$.\medskip
    
    \textbf{Step 5:} Finally, we show that $\gamma_C\geq 0$. For this, let $a,b,c$ denote three distinct alternatives with $a \ssucc b \ssucc c$. Moreover, we partition the voters in three disjoint sets $I^1$, $I^2$, and $I^3$ with $|I^1|=|I^2|=\frac{n}{2}-1$ and $|I^3|=2$. Since $n\geq 4$, none of these sets is empty. Now, consider the following three profiles, where $I^3=\{i,j\}$.
    
    \begin{profile}{L{0.06\columnwidth} L{0.2\columnwidth} L{0.2\columnwidth} L{0.18\columnwidth} L{0.18\columnwidth}}
    $R^1$: & $I^1: a,b,c\dots$ & $I^2$: $b,a,c\dots$ & $I^3$: $c,a,b$\\
    $R^2$: & $I^1: a,b,c\dots$ & $I^2$: $b,a,c\dots$ & $i$: $c,a,b$ & $j$: $c,b,a$ \\
    $R^3$: & $I^1: a,b,c\dots$ & $I^2$: $b,a,c\dots$ & $I^3$: $c,b,a$
    \end{profile}
    
    All three profiles are in $\mathcal{D}_C^\ssucc$: in $(R^1,\ssucc)$ and $(R^2,\ssucc)$, $a$ is the Condorcet winner, and $b$ is the Condorcet winner in $(R^3,\ssucc)$. Next, observe that two applications of non-perversity show that $f(R^1,a)\geq f(R^3,a)$. On the other hand, we have that $f(R^1,a)=\gamma_C+\sum_{i\in I^1}\gamma_i$ and $f(R^3,a)=\sum_{i\in I^1}\gamma_i$. Hence, it follows immediately that $\gamma_C\geq 0$. This completes the proof of Claim (1) because we already know that $f(R)=\gamma_C\mathit{COND}^\ssucc(R)+\sum_{i\in N}\gamma_id_i(R)$ for all profiles $R\in\mathcal{D}_C^{\ssucc,x}$ and alternatives $x\in A$ and $\mathcal{D}_C^\ssucc=\bigcup_{x\in A} \mathcal{D}_C^{\ssucc,x}$.\medskip
    
    \textbf{Proof of Claim (2)}: For the second claim, suppose that $\mathcal{D}_C^\ssucc\subsetneq\mathcal{D}$. It is obvious that random dictatorships are strategyproof and non-imposing on $\mathcal{D}$. We thus focus on the converse claim and let $f$ denote a strategyproof and non-imposing SDS on $\mathcal{D}$. First, by the reasoning before the proof of Claim (1), it follows that $f$ is \emph{ex post} efficient. Consequently, $f$ is strategyproof and non-imposing on $\mathcal{D}_C^\ssucc$. We can therefore use Claim (1) to infer that there are values $\gamma_C\geq 0$ and $\gamma_i\geq 0$ for all $i\in N$ such that $f(R)=\gamma_C\mathit{COND}^\ssucc(R)+\sum_{i\in N} \gamma_i d_i(R)$ for all $R\in\mathcal{D}_C^\ssucc$. We thus have to show that $\gamma_C=0$ and to extend this representation to the larger domain $\mathcal{D}$.
    
    For this, consider a profile $R\in\mathcal{D}\setminus \mathcal{D}_C^\ssucc$ and an alternative $x\in A$. We define $B=\{y\in A\colon g_{(R, \ssucc)}(y,x)> 0\}$ and note that $B\neq \emptyset$ because there is no Condorcet winner in $(R,\ssucc)$. Let $y$ denote the maximal alternative in $B$ according to $\ssucc$, i.e., $y\ssucc z$ for all $z\in B\setminus \{y\}$. Moreover, we define $I=\{i\in N\colon y\succ_i x\}$ as the set of voters who prefer $y$ to $x$ in $R$ and note that $|I|\geq \frac{n}{2}$. Next, consider the profile $R'$ derived from $R$ by letting all voters $i\in I$ make $y$ into their best alternative and all voters $i\in N\setminus I$ move up $y$ until it is directly below $x$. If $|I|>\frac{n}{2}$, then $y$ is the Condorcet winner in $(R',\ssucc)$ and $R'\in\mathcal{D}_C^\ssucc$. If $|I|=\frac{n}{2}$, it must hold that $y\ssucc x$ as otherwise $g_{(R,\ssucc)}(y,x)<0$. We claim also in this case that $y$ is the Condorcet winner in $(R',\ssucc)$. If this was not the case, there is an alternative $z$ such that $g_{(R',\ssucc)}(z,y)>0$. Since all voters $i\in I$ top-rank $y$ in $R'$, this means that $z\succ_i' y$ for all $i\in N\setminus I$ and $z\ssucc y$. Moreover, because all voters $i\in N\setminus I$ rank $y$ directly below $x$, it follows that $z\succ_i x$ for all these voters. However, we then have that $z\in B$ and $z\ssucc y$, which contradicts the the definition of $y$. This proves that $R'\in\mathcal{D}_C^\ssucc$, and we can now use the same analysis as for proof of Claim (2) in \Cref{thm:cond} to complete the proof.
\end{proof}

As last result, we discuss the proof of \Cref{thm:groupCond}.

\groupCond*
\begin{proof}
For the direction from right to left of both claims, we refer to the proof sketch in the main body as it discusses in detail why the Condorcet rule and dictatorships satisfy the required axioms. We focus here on the direction from left to right and consider therefore a group-strategyproof and non-imposing SDS $f$ on a domain $\mathcal{D}$ with $\mathcal{D}_C\subseteq\mathcal{D}$. First, we show by contradiction that $f$ is \emph{ex post} efficient. Hence, assume that there are alternatives $x,y\in A$ and a profile $R\in\mathcal{D}$ such that $f(R,y)>0$ even though $x$ Pareto-dominates $y$ in $R$. On the other hand, there is a profile $R'$ such that $f(R',x)=1$ by non-imposition. It is now easy to see that the set of all voters $N$ can group-manipulate by deviating from $R$ to $R'$ because all voters prefer $x$ to $y$. The initial assumption is therefore wrong and $f$ satisfies \emph{ex post} efficiency.

Since group-strategyproofness implies strategyproofness, we can now invoke \Cref{lem:DCx} to derive that for every alternative $x\in A$, there are values $\gamma_C^x$ and $\gamma_i^x\geq 0$ for all $i\in N$ such that $f(R)=\gamma_C^x\mathit{COND}(R)+\sum_{i\in N} \gamma_i^x d_i(R)$ for all $R\in\mathcal{D}_C^x$. We show next that $\gamma_i^x=\gamma_i^y$ for all voters $i\in N$ and alternatives $x,y\in A$. For this, consider the profiles $R^1$ and $R^2$ shown below, where $i$ is an arbitrary voter and $a,b,c$ denote three distinct alternatives.

\begin{profile}{L{0.1\columnwidth} L{0.3\columnwidth} L{0.3\columnwidth}}
		$R^1$: & $i$: $c,a,b,\dots$ & $N\setminus \{i\}$: $a,b,c,\dots$\\
		$R^2$: & $i$: $c,a,b,\dots$ & $N\setminus \{i\}$: $b,a,c,\dots$
	\end{profile}

Since $n\geq 3$, it holds that $a$ is the Condorcet winner in $R^1$ and $b$ in $R^2$. This entails that $f(R^1,c)=\gamma_i^a$ and $f(R^2,c)=\gamma_i^b$. Finally, group-strategyproofness implies that $f(R^1, \{a,b\})=f(R^2,\{a,b\})$ because otherwise, the group $N\setminus \{i\}$ can group-manipulate by deviating from $R^1$ to $R^2$ or \emph{vice versa}. For instance, if $f(R^1, \{a,b\})<f(R^2, \{a,b\})$, then $f(R^1)\not\succsim_j^\sd f(R^2)$ for all $j\in N\setminus \{i\}$, which proves that this is indeed a group-manipulation. Since all other alternatives are Pareto-dominated, we thus infer that $\gamma_i^a=f(R^1,c)=f(R^2,c)=\gamma_i^b$. Analogous to the proof of \Cref{thm:cond}, it follows now that $\gamma_C^a=\gamma_C^b$ and there are thus $\gamma_C$ and $\gamma_i\geq 0$ for $i\in N$ such that $f(R)=\gamma_C\mathit{COND}(R)+\sum_{i\in N}\gamma_i d_i(R)$ for all $R\in\mathcal{D}_C$.

Next, we show that if $\gamma_i>0$, then $\gamma_i=1$. Assume for contradiction that this is not the case, i.e., there is a voter $i\in N$ with $\gamma_i>0$ and $\gamma_i\neq 1$. To derive a contradiction, we consider the profiles $R^3$ and $R^4$ shown below.

\begin{profile}{L{0.1\columnwidth} L{0.3\columnwidth} L{0.3\columnwidth}}
		$R^3$: & $i$: $c,a,b,\dots$ & $N\setminus \{i\}$: $b,a,c,\dots$\\
		$R^4$: & $i$: $a,b,c,\dots$ & $N\setminus \{i\}$: $a,b,c,\dots$
	\end{profile}

Since $b$ is the Condorcet winner in $R^3$ and $0<\gamma_i$, we have that $f(R^3,c)=\gamma_i>0$. In particular, this implies that $\gamma_i\leq 1$ as $f$ is otherwise not well-defined. In turn, our contradiction assumption entails that $\gamma_i<1$ and therefore $f(R^3,b)=1-f(R^3,c)>0$ because $R^3\in\mathcal{D}_C$. On the other hand, \emph{ex post} efficiency shows that $f(R^4,a)=1$. However, this means that the group of all voters can group-manipulate by deviating from $R^3$ to $R^4$ because $f(R^3,U(\succ^3_j,a))<1=f(R^4,U(\succ^3_j,a))$ for all voters $i\in N$. This contradicts that $f$ is group-strategyproof and thus, if $\gamma_i>0$, then $\gamma_i=1$. Furthermore, it is not possible that $\gamma_i=1$ and $\gamma_j=1$ for distinct voters $i,j$. This follows by considering the profile $R^5\in\mathcal{D}_C$ shown below: if both $\gamma_i=1$ and $\gamma_j=1$, then $f(R^5,b)=f(R^5,c)=1$ which violates the definition of an SDS.

\begin{profile}{L{0.1\columnwidth} L{0.3\columnwidth} L{0.2\columnwidth} L{0.2\columnwidth}}
		$R^5$: & $N\setminus \{i,j\}$: $a,b,c,\dots$ & $i$: $b,a,c,\dots$ & $j$: $c,a,b,\dots$
	\end{profile}

As a consequence, we infer for all profiles $R\in\mathcal{D}_C$ that either $f(R)=d_i(R)$ for some $i\in N$ or $f(R)=\mathit{COND}(R)$. This proves Claim (1): only dictatorial SDSs and the Condorcet rule are group-strategyproof and non-imposing on the Condorcet domain. 

For proving the second claim, we assume that there is a profile $R^*\in\mathcal{D}$ such that for every alternative $x\in A$, there is another alternative $y\in A\setminus \{x\}$ such that $g_{R^*}(y,x)>0$. Now, let $a$ denote an alternative such that $f(R^*,a)>0$. Moreover, let $b$ denote an alternative with $g_{R^*}(b,a)>0$ and let $I$ denote the set of voters with $b\succ^*_i a$. We consider the profile $R'$ derived from $R^*$ by letting all voters $i\in I$ make $b$ into their best alternative. Clearly, $b$ is the Condorcet winner in $R'$ because $|I|>\frac{n}{2}$, which entails that $R'\in\mathcal{D}$. If $f(R')=\mathit{COND}(R')$, the voters $i\in I$ can group-manipulate by deviating from $R$ to $R'$ because they all prefer $b$ to $a$. Hence, $f(R)$ is not the Condorcet rule for profiles on $\mathcal{D}_C$, which means that there is a voter $i\in N$ such that $f(R)=d_i(R)$ for all $R\in\mathcal{D}_C$. 

For completing the proof, we need to show that voter $i$ dictates the outcome in all profiles. For doing so, consider an arbitrary profile $R\in\mathcal{D}$ and let $x$ denote voter $i$'s favorite alternative in $R$. We suppose for contradiction that $f(R,x)<1$. Now, consider the profile $R'$ in which all voters $j\in N\setminus \{i\}$ prefer an alternative $y\in A\setminus \{x\}$ the most and $x$ the least, and voter $i$ reports $\succ_i$. Clearly, $y$ is the Condorcet winner in $R'$ and thus $f(R',x)=d_i(R',x)=1$. Now, it is easy to see that if $f(R,x)<1$, then the voters in $N\setminus \{i\}$ can group-manipulate by deviating from $R'$ to $R$ because they prefer every other lottery to $f(R')$. Hence, the assumption that $f(R,x)<1$ contradicts the group-strategyproofness of $f$. This proves that $f(R)=d_i(R)$ for all $R\in\mathcal{D}$ and thus proves Claim (2).
\end{proof}
\end{document}